\documentclass[11pt]{article}
\usepackage[margin=1in]{geometry}
\usepackage[T1]{fontenc}
\usepackage{tgtermes}
\usepackage{setspace}
\usepackage{amsmath,amsfonts,amssymb,mathtools}
\usepackage{amsthm,thmtools}
\declaretheoremstyle[headfont=\normalfont\rmfamily\bfseries\scshape]{ORhead}

\usepackage{algorithm}
\usepackage{algpseudocode}
\usepackage{tikz}

\usepackage{natbib}
 \bibpunct[, ]{(}{)}{,}{a}{}{,}%
\setcitestyle{acmnumeric}

\usepackage{graphicx}
\usepackage{array}
\usepackage{multirow, hhline}
\usepackage{paralist}

\usepackage{algpseudocode}
\usepackage{xcolor}
\usepackage{wrapfig}
\usepackage[english]{babel}		
\usepackage{comment}
\usepackage{chngpage}			
\usepackage{color}
\usepackage{blindtext}
\usepackage{multicol}
\usepackage{color}

\usepackage{pgfplots}
\usepackage{pgfgantt}
\usepackage{tikz}
\usetikzlibrary{arrows}
\usetikzlibrary{calc}
\usepackage{algpseudocode}
\usepackage{subcaption}
\usepackage{tabto}
\usepackage{csquotes} 
\usepackage{multirow}

\usepackage{hyperref}
\usepackage[capitalise]{cleveref} 
\newcommand\pcref[1]{(\cref{#1})}

\crefname{enumi}{}{}

\declaretheorem[style=ORhead]{theorem}
\declaretheorem[style=ORhead,numbered=no,name=Theorem]{theorem*}
\declaretheorem[style=ORhead]{lemma}

\declaretheorem[style=ORhead]{definition}
\declaretheorem[style=ORhead]{claim}
\declaretheorem[style=ORhead]{fact}
\declaretheorem[numbered=no]{remark}

\newcommand{\argmax}{{\operatorname{\mathrm{arg\,max}}}}

\newcommand{\R}{\mathbb{R}}

\DeclarePairedDelimiterX{\inp}[2]{\langle}{\rangle}{#1, #2}
 
\DeclarePairedDelimiterX{\cb}[1]{\{}{\}}{#1}

\newcommand{\RR}{\mathbb{R}}

\newcommand{\infd}{pole}

\spacing{1.5}

\title{\fontfamily{qhv}\selectfont\bfseries Competitive Equilibrium for Chores: from Dual Eisenberg-Gale to a Fast, Greedy, LP-based Algorithm}
\author{Bhaskar Ray Chaudhury\thanks{University of Illinois at Urbana-Champaign, USA} \and Christian Kroer\thanks{Columbia University, USA} \and Ruta Mehta\footnotemark[1] \and Tianlong Nan\footnotemark[2]}
\date{}

\begin{document}

\maketitle

\paragraph{{\fontfamily{qhv}\selectfont\bfseries  Abstract.}}
We study the computation of competitive equilibrium for Fisher markets with $n$ agents and $m$ divisible chores. Competitive equilibria for chores are known to correspond to the nonzero KKT points of a program that \emph{minimizes} the product of agent disutilities, which is a non-convex program whose zero points foil iterative optimization methods.
We introduce a dual-like analogue of this program, and show that a simple modification to our ``dual'' program avoids such zero points, while retaining the correspondence between KKT points and competitive equilibria. This allows, for the first time ever, application of iterative optimization methods over a convex region for computing competitive equilibria for chores. We next introduce a greedy Frank-Wolfe algorithm for optimization over our program and show a new state-of-the-art convergence rate to competitive equilibrium. Moreover, our method is significantly simpler than prior methods: each iteration of our method only requires solving a simple linear program. We show through numerical experiments that our method is extremely practical: it easily solves every instance we tried, including instances with hundreds of agents and up to 1000 chores, usually in 10-30 iterations, is simple to implement, and has no numerical issues.

\section{Introduction} 

We study competitive equilibrium (CE) for Fisher markets with chores, where a set of divisible chores need to be allocated among a set of agents. Each agent incurs \emph{disutility} from the chores assigned to her, and is compensated with payments for the chores she undertakes. We study the most classical utility model, where the disutility of each agent is linear in the chores assigned to her. Each agent is assumed to have a fixed earning requirement.
CE is a central solution concept in microeconomics which specifies how markets can assign prices and allocations in a way that satisfies demand and supply.
In the case of Fisher markets with chores, prices and allocations form a CE if (i) each agent gets her disutility-minimizing bundle that achieves her earning requirement, and (ii) all chores are fully allocated. 

A classical motivation for the Fisher market model with goods, is the \emph{competitive equilibrium with equal incomes (CEEI)}, which is a mechanism for fairly and efficiently dividing a set of goods among a set of agents~\citep{varian1974equity}.
The CEEI idea naturally extends to the case of chores, and still provides similar fairness and efficiency properties.
Examples of fair chore allocation problems include family members dividing household chores, workers dividing job shifts, faculty dividing teaching loads, peer review assignment, and so on.


In a seminal paper, \citet{BogomolnaiaMSY17} showed that CE exists in chores Fisher markets with 1-homogeneous convex disutility functions\footnote{\citet{BogomolnaiaMSY17} measure \emph{utility}, in which case the utility function is concave; for our purposes it is easier to work with the equivalent \emph{disutility} of an agent, in which case the corresponding function is convex.}, and gave an optimization characterization of the set of CE.
They also showed that, despite the resemblance to CE with goods, CE with chores has surprisingly different structural and computational properties. Firstly, the set of CE can be disconnected~\citep{BogomolnaiaMSY17}, in contrast to CE with goods where all equilibria form a convex set. 
Secondly, the CE set for goods is the set of optimal solutions to a convex program known as the \emph{Eisenberg-Gale} convex program, which is a program that maximizes the product of agent utilities.
In the case of chores, the set of CE is still captured by an analogous program which attempts to \emph{minimize} the product of agent \emph{disutilities}.
\citet{BogomolnaiaMSY17} show that every KKT point of this program (call it the EG program for chores) where every agent receives strictly positive distuility corresponds to a CE. Importantly, KKT points where some agent receives zero disutility do \emph{not} correspond to CE, yet such points are the global minimizes of the program. This effectively means that EG for chores has an open constraint, and the global infimum is achieved at the closure of this open constraint. This is highly problematic from an optimization perspective, because these zero points may attract iterative optimization methods.
Due to this algorithmic difficulty, polynomial-time algorithms are only known for computing approximate CEEI~\citep{boodaghians2022polynomial, ChaudhuryGMM22}, and existing methods all rely on highly specialized subroutines that are not very practical.

As is the case for the EG program for goods, we will work with the logarithm of the 
product
of agent disutilities, which converts the 
product
of disutilities into a concave function, albeit one that we need to minimize.
Note that the domain of the logarithm function automatically rules out allocations with zero disutility for some agent, and thus all KKT points of this objective correspond to CE. 
However, this does not solve the problem of applying iterative optimization methods: under this formulation the objective has infeasible points on the boundary, corresponding to zero disutility allocations, such that the objective converges to negative infinity as we approach these points. We refer to these points as 
\emph{\infd}s, in analogy to their usage in complex analysis. 
In this formulation, starting from almost every feasible points, standard iterative optimization methods move towards these {\infd}s and the objective tends to negative infinity~\citep{boodaghians2022polynomial}. 

To circumvent the poles issue, sophisticated and highly specialized iterative methods have been developed~\citep{boodaghians2022polynomial, ChaudhuryGMM22} for finding a nonzero KKT point. However, both methods require exactly solving a non-linear convex program at each iteration -- the projection subroutine in~\citep{boodaghians2022polynomial}, and the balanced flow subroutine in~\citep{ChaudhuryGMM22}. 
Exactly solving a nonlinear convex program at every iteration can be expensive, and thus these methods are not practical for large-scale problems.
A major roadblock for developing more practical algorithms is the presence of poles in the EG for chores program. 
This opens up a natural question: 
\begin{quote}
\emph{
Does there exists an optimization formulation of the chores problem that avoids the poles issue?
}
\end{quote}
An affirmative answer to this question would open up the CE for chores problem to more standard iterative optimization methods, making it more tractable both theoretically and practically.
As we discuss below, we show that such a formulation indeed exists, and it leads to significantly faster algorithms.

\paragraph{Organization.} The remainder of this section discusses our contributions and related work. \cref{Sec:prelim} gives preliminaries for Fisher markets with chores and existing mathematical programs. \cref{Sec:convex-pgm} introduces two new (dual EG) programs capturing CE with chores. \cref{Sec:Alg} provides with a new greedy, fast, LP-based algorithm to compute a CE, and proves its convergence properties. \cref{sec:experiments} demonstrates practical efficiency of our new algorithm and~\cref{sec:conclusion-and-discussion} includes conclusion and discussion.

\subsection{Our Contributions}

\paragraph{Novel program with no poles.} As our first main result, we give a novel program that maximizes a convex function subject to simple linear constraints. We show that our program circumvents the issue of poles of the EG program for chores, while maintaining the one-to-one correspondence between CE and  KKT points. To the best of our knowledge, this is the first program that makes the chores problem amenable to standard gradient-descent-type algorithmic approaches.



The inspiration for our new program comes from the dual of the EG program in the context of goods. The dual EG program for goods has been used extensively for deriving effective algorithms, as we discuss in \cref{sec:related-work}. 
In a formal sense, there is no duality theory for the chores setting, since we are dealing with nonconvex programs.
Nonetheless, we develop a natural ``dual'' formulation of the EG program for chores. 
However, just as the primal EG program for chores, the {\em natural} dual EG program for chores also has poles. 
Surprisingly, we show that the poles of our new program can be removed by adding a linear constraint. This constraint not only gets rid of all \infd s, but also preserves all KKT points without introducing any new KKT points. 
We call the new program \eqref{chores dual redundant}. 
We observe that, although a non-convex optimization problem, \eqref{chores dual redundant} has strong duality-like structural properties: 
its objective {\em matches} that of the primal EG program at the KKT points. 


\paragraph{Greedy Frank Wolfe algorithm.} Armed with this new program, it is natural to attempt a first-order method to compute a CE with chores. In this paper, we adopt the greedy variant of the famous \emph{Frank-Wolfe} (FW) method~\citep{frank1956algorithm}. 
In~\cref{Sec:Alg}, we show that this method converges to an $\epsilon$-approximate CE (formally defined in Definition~\eqref{def:appx-CEEI}) in $\mathcal{\tilde{O}}(\frac{n}{\epsilon^2})$ iterations, 
matching the convergence guarantees in~\citet{boodaghians2022polynomial, ChaudhuryGMM22}, in terms of dependence on $\epsilon$, while improving the dependence on the number of buyers and chores. 
To be specific, the exterior point method by~\citet{boodaghians2022polynomial} requires $\mathcal{\tilde{O}}(\frac{n^3}{\epsilon^2})$ iterations, and the combinatorial algorithm by~\citet{ChaudhuryGMM22} requires $\mathcal{\tilde{O}}(\frac{nm}{\epsilon^2})$ iterations.
We emphasize that, unlike the algorithms in~\citet{boodaghians2022polynomial, ChaudhuryGMM22},  an iteration of GFW involves solving only a single LP, as opposed to more general convex programs in prior methods. 
Further, the dependence on problem size in the bound on the number of iterations for our GFW method is significantly smaller than prior methods~\citep{boodaghians2022polynomial, ChaudhuryGMM22}. 
Combining these facts, we expect GFW to be considerably faster in practice, especially for large instances. We confirm this intuition through a comprehensive set of empirical results.

\paragraph{Convergence for relative strongly convex maximization.}
Consider a problem of maximizing a relative-strongly-convex function $f$ (w.r.t. some reference function $h$) over a nonempty closed convex set. 
We show that after $T \in \mathcal{O}\Big(\frac{f^* - f^{0}}{\alpha\epsilon}\Big)$\footnote{Here, $f^*$ and $f^0$ are the optimal and initial objective values, respectively. $\alpha$ denotes the modulus of the relatively strong convexity.} iterations of the GFW algorithm, there exists an iterate $k \in [T]$, such that  the distance
between $x_k$ and $x_{k-1}$ is small compared to $\epsilon$, as measured by the Bregman divergence induced by the reference function. 
See~\cref{sec:general} for our generalized results. 
We use this result, along with a particular choice of $h$, to show that GFW finds an $\epsilon$-approximate CE. 
We believe that this general approach could lead to interesting guarantees for other problems of maximizing a convex function (e.g. the approximate $\ell_0$-penalized PCA\footnote{Principal component analysis.} problem~\citep{luss2013conditional}, multiplicative problems~\citep{benson1995concave}, etc.). 
In the appendix, we also present an alternative primal-dual analysis of GFW specifically for the \eqref{chores dual redundant} program, with mildly stronger guarantees (see~\cref{sec:primal-dual-convergence}). 

\paragraph{Empirical evaluation.} 
We compare our GFW method to the state-of-the-art existing method: the exterior-point method (EPM) from \citet{boodaghians2022polynomial}, on a large number of instances sampled from five different classes of random disutility distributions, as well as on a semi-synthetic dataset derived from a conference review bidding dataset.
Our GFW algorithm solves every instance that we tried it on, usually in 10-30 iterations, and in less than a minute for instances with several hundred buyers and chores. In contrast, the prior EPM fails on a large number of instances (for some disutility distributions it fails on every instance once the number of buyers and chores is on the order of a few hundred each), though it is usually competitive with our method in the cases where it does not fail.
We conclude that our GFW method is the first highly practical method: it is extremely robust and simple to implement, and can computes an exact CE within a minute for every instance we tried, even with upwards of 300 buyers and chores.

\subsection{Related Work}
\label{sec:related-work}

CE with goods has seen a long line of research since the 1950s. Early work~\citep{ArrowD54} showed the existence of CE, and also formulated convex programs~\citep{eisenberg1959consensus} that capture all CE under specific utility functions. In particular,~\citet{eisenberg1959consensus} show that any allocation that maximizes the product of the utility functions of the agents corresponds to a CE. Despite the existence of the convex program, and polynomial bounded rationality of the optimal solutions, there has been a long study on a variety of algorithms to compute CE under linear utilities, including interior point methods~\citep{jain2007polynomial, Ye08}, combinatorial methods~\citep{devanur2008market, Orlin10, DuanM15, GV19, ChaudhuryM18}, and first-order methods and dynamics~\citep{birnbaum2011distributed, zhang2011proportional,gao2020first}. One of the reasons for designing special-purpose algorithms is attributed to the lack of practically fast convex program solvers for large instances~\citep{Vazirani12}-- for instance, convex program solvers are significantly slower than LP solvers, and are practical only for medium-size instances (see benchmarks in Mittelmann). Furthermore, the broad algorithmic literature has also revealed several structural and economically valuable insights on CE~\citep{DuanGM16,DevanurGV13,GV19}.  
Finally, CE has been used in a variety of applications, especially in the case of linear utility or disutility functions~\citep{goldman2015spliddit,conitzer2022pacing,allouah2023fair}.
For example, the online platform Spliddit has seen tens of thousands of users for the past few years~\citep{goldman2015spliddit}. Spliddit uses additive preferences which is the parallel of linear utility functions in the context of sharing indivisible items. CE with linear utilities has also seen extensive use in online fair allocation~\citep{azar2016allocate,banerjee2022online,gao2021online,benade2023fair}.

The problem of computing a CE involving chores has attracted significant attention since the characterization result of \citet{BogomolnaiaMSY17}. 
There are polynomial-time algorithms if the number of agents or the number of chores is a constant: in this special case,~\citet{BS23} showed CE with chores can be computed in strongly polynomial time and \citet{GargM20} showed a polynomial-time algorithm for computing CE in settings that include both goods and chores.
\citet{chaudhury2021competitive} developed a simplex-like algorithm based on Lemke’s scheme for additively separable piecewise linear concave utilities, and showed that it can be implemented for small instances with 10-20 agents and chores. 
As mentioned previously~\citet{boodaghians2022polynomial} and \citet{ChaudhuryGMM22} developed polynomial-time algorithms for computing approximate CE with chores, where their iteratations require solving non-linear convex programs. 

\section{Preliminaries}
\label{Sec:prelim}


\subsection{Fisher Market with Chores} 

Analogous to \emph{Fisher markets} in the goods setting, a chores market comprises of a set of $n$ agents, and a set of $m$ \emph{divisible} chores. Each agent $i$ has a disutility of $d_{ij}$ for one unit of chore $j$. Given a \emph{bundle} of chores $x_i \in \mathbb{R}^m_{\geq 0}$, where $x_{ij}$ denotes the amount of chore $j$, agent $i$'s total linear disutility for bundle $x_i$ is $\sum_{j \in [m]} d_{ij} x_{ij}$. We also write $\sum_{j \in [m]} d_{ij} x_{ij}$ as $\langle d_i, x_i \rangle$, where $d_i = (d_{i1}, d_{i2}, \dots, d_{im})$. 
Further, each agent $i$ has an earning requirement of $B_i \geq 0$ units of money. 

In a CE for a chores market, we wish to find a price $p_j$ for each chore $j$, and an allocation $x \in \mathbb{R}^{n \times m}_{\geq 0}$ of chores to the agents, where each $x_{ij}$ is the amount of chore $j$ allocated to agent $i$, such that (i) each agent $i$ gets a disutility minimizing bundle $x_i = (x_{i1}, x_{i2}, \dots, x_{im})$, subject to her earning requirement of $B_i$ units of money, and (ii) all chores are allocated. Formally,

\begin{definition} 
   \label{Def:CE}
    A price vector $p \in \mathbb{R}^m_{\geq 0}$, and an allocation $x \in \mathbb{R}^{n \times m}_{\geq 0}$ satisfy \emph{competitive equilibrium} (CE) if and only if
    \begin{itemize} 
        \item[(E1)] $\inp{p}{x_i} = B_i$ for all $i \in [n]$ \label{con:Def-CE-1}; 
        \item[(E2)] $\inp{d_i}{x_i} \leq \inp{d_i}{y_i}$ for all $y_i$ such that $\inp{p}{y_i} \geq \inp{p}{x_i}$, for all $i \in [n]$ \label{con:Def-CE-2}; 
        \item[(E3)] $\sum_i x_{ij} = 1$ for all $j \in [m]$ \label{con:Def-CE-3}. 
    \end{itemize}
\end{definition} 
Condition (E1) and (E2) ensure that each agent gets a disutility-minimizing bundle subject to her earning constraint, and (E3) captures that all chores are allocated in this process. CEEI is one of the most well-studied special cases of CE, where $B_i = 1$, for all $i \in [n]$. It is a well known (see~\citet[Lemma 1]{BogomolnaiaMSY17}) that if $(p,x)$ satisfy CEEI, then (i) $x$ is \emph{envy-free}, i.e., $d_i(x_i) \leq d_i(x_j)$ for all $i, j \in [n]$, and (ii) $x$ is \emph{Pareto-optimal}, i.e., there is no $y$ such that $d_i(y_i) \leq d_i(x_i)$ with at least one strict inequality. CEEI satisfies many other desirable properties like \emph{core-stability}, and we refer the reader to~\citet{BogomolnaiaMSY17} for a more detailed explanation of the remarkable fairness and efficiency properties satisfied by CEEI. The existence of CE in the chores market can be shown by adapting the fixed point arguments in~\citet{ArrowD54}. 

Our algorithms will be concerned with computing approximate CE solutions. To that end, we relax conditions (E1), (E2), and (E3) in~\cref{Def:CE} as follows,

\begin{definition}[$\epsilon$-approximate CE~\citep{boodaghians2022polynomial}] 
   \label{def:appx-CEEI}
A price vector $p \in \mathbb{R}^m_{\geq 0}$, and an allocation $x \in \mathbb{R}^{n \times m}_{\geq 0}$ satisfy $\epsilon$-\emph{competitive equilibrium} (CE) if and only if
    \begin{itemize} 
        \item[(1)] $(1 - \epsilon) B_i \leq \inp{p}{x_i} \leq \frac{1}{1 - \epsilon} B_i$ for all $i \in [n]$; 
        \item[(2)] $(1 - \epsilon) \inp{d_i}{x_i} \leq \inp{d_i}{y_i}$ for all $y_i$ such that $\inp{p}{y_i} \geq \inp{p}{x_i}$, for all $i \in [n]$; 
        \item[(3)] $1 - \epsilon \leq \sum_i x_{ij} \leq \frac{1}{1 - \epsilon}$ for all $j \in [m]$. 
    \end{itemize}
\end{definition} 
$\epsilon$-approximate CEEI is also analogously defined. Furthermore, we consider a stronger notion of approximate CE, called \emph{$\epsilon$-strongly approximate CE}. 

\begin{definition}\label{def:sappx-CEEI} A pair of $(p, x)$ is said to be an $\epsilon$-strongly approximate CE if $(E2)$ and $(E3)$ holds strictly and only $(E1)$ is relaxed to $(1)$. 
\end{definition}

\subsection{Mathematical Programming Approaches to Competitive Equilibria}

In the case of CE for goods, it is known that a CE allocation can be obtained from the Eisenberg-Gale convex program. In the goods case with linear utilities, each buyer has a value vector $v_i \in \R_+^m$ describing their value per good, and they compute demands by maximizing their obtained value subject to budget constraints. The EG program for this case is stated on the left below.
\begin{equation}
    \left.
    \begin{aligned}
        \max_{x \geq 0} \; &\sum_{i \in [n]} B_i \log\bigg(\sum_{j \in [m]} v_{ij}x_{ij} \bigg) \\
        \text{ s.t. } & \sum_{i \in [n]} x_{ij} =1 \quad  \forall j \in [m]. 
    \end{aligned}
    \hspace{30pt} 
    \right|
    \hspace{30pt}
    \begin{aligned}
        \min_{\beta \geq 0, p \geq 0} &\sum_j p_j - \sum_i B_i \log(\beta_i) \\
        \text{ s.t. } & p_j \geq \beta_i v_{ij}\quad  \forall i,j.
        \label{Pgm:EG-goods-dual}
    \end{aligned} 
\end{equation}
On the right in~\eqref{Pgm:EG-goods-dual} is the dual of the EG program, 
which minimizes a function of $p$ and $\beta$, where $p$ and $\beta$ can be interpreted as prices for the goods, and inverse bang-per-buck variables for the buyers, respectively.
We will use the dual of the goods EG program as inspiration for our chores approach.

\citet{BogomolnaiaMSY17} show an interesting characterization of all CE in the chores market. Similar to the case of goods,  they show that for chores there is also an EG-style program for computing a CE.
More specifically, they show that there is a one-to-one correspondence between the set of CE and all KKT points of the following program,
\begin{equation}
    \begin{aligned}
        \inf_{x \geq 0} \; & \prod_{i \in [n]} \bigg(\sum_{j \in [m]} d_{ij}x_{ij} \bigg)^{B_i} \\
        \text{ s.t. } & \sum_{i \in [n]} x_{ij} =1 \quad  \;\;\;\;\; \forall\, j \in [m] \\
        &\sum_{j \in [m]} d_{ij}x_{ij} > 0 \quad \forall\, i \in [n].
    \end{aligned}
    \label{Pgm:EG-Chores-prod}
\end{equation}
Applying the logarithm, which is a monotone transformation, to the objective yields a mathematical program analogous to the EG program for goods,
\begin{equation}
    \begin{aligned}
        \inf_{x \geq 0} \; &\sum_{i \in [n]} B_i \cdot \log\bigg(\sum_{j \in [m]} d_{ij}x_{ij} \bigg) \\
        \text{ s.t. } & \sum_{i \in [n]} x_{ij} =1 \quad \;\;\;\;\;  \forall\, j \in [m].
    \end{aligned}
    \label{Pgm:EG-Chores-sumlog}
\end{equation}
After applying the logarithm, we no longer need an explicit open constraint, since the domain constraint on the logarithm already enforces strictly positive disutilities.

However, this program has several disadvantages compared to the goods case. Most notably, it is a non-convex program, 
and the objective function can tend to negative infinity within the feasible region (as the disutility of an agent tends to zero).
This makes it challenging algorithmically, because iterative methods that operate inside the feasible set are attracted to these ``negatively-infinite-objective'' points \citep{boodaghians2022polynomial}.



\subsection{Bregman Divergence and Generalized Strong Convexity} 

Let $h: \Omega \rightarrow \mathbb{R}$ be a differentiable, strictly convex function defined on a convex set $\Omega$. 
The Bregman divergence associated with $h$ between $y, x \in \Omega$ is defined as 
$D_h(y \,\Vert\, x) = h(y) - h(x) - \inp{\nabla h(x)}{y - x}$. 
The function $h$ is called the reference function of $D_h(\cdot\,\Vert\,\cdot)$. 
For example, if $h(x) = \frac{1}{2} {\left\Vert x \right\Vert}^2$ then 
$D_h(y\,\Vert\,x) = \frac{1}{2} {\left\lVert y - x \right\Vert}^2$. 
Since $h$ is strictly convex, 
$D_h(y \,\Vert\, x) \geq 0$ for all $y, x \in \Omega$, and $D_h(y \,\Vert\, x) = 0$ iff $y = x$.  

We say a differentiable function $f$ is \emph{$\alpha$-strongly convex with respect to $h$} if $\text{dom}\,f \subseteq \Omega$ and 
\begin{equation}
     f(y) \geq f(x) + \inp{\nabla f(x)}{y - x} + \alpha D_h(y\,\Vert\,x) \quad \forall\, y, x \in \text{dom}\,f. 
    \label{eq:generalized-strong-convexity}
\end{equation} 
We refer~\cref{eq:generalized-strong-convexity} with an arbitrary $D_h(\cdot\,\Vert\,\cdot)$ as \emph{generalized strong convexity} or \emph{relatively strong convexity}, and it reduces to (standard) strong convexity if $D_h(y\,\Vert\,x) = \frac{1}{2} {\left\lVert y - x \right\Vert}^2$. 

\section{The Chores Dual Program}
\label{Sec:convex-pgm}

In this Section, we introduce a novel program with no \infd s, such that there is a one-to-one correspondence between CE in the chores market and the KKT points of our program. Our inspiration for this program comes from the dual program for CE in the goods setting (\cref{Pgm:EG-goods-dual}).
In the goods case, the dual program has been used extensively in order derive various algorithmic approaches for computing CE. The \emph{proportional response dynamics} can be derived by a change of variable and dualization step on the goods EG dual~\citep{zhang2011proportional,birnbaum2011distributed}, and the \emph{PACE} algorithm for online Fisher market equilibrium is also derived from this program~\citep{gao2021online,liao2022nonstationary}.
Given the algorithmic usefulness of the EG dual in the goods case, it is natural to ask whether there is an analogue for the chores case. The EG program for chores, however, is not a convex program, and thus convex duality does not let us directly derive a dual.
Nonetheless, we may \emph{guess} that the chores EG ``dual'' should be the maximization version of the goods EG dual, analogously to how the chores ``primal'' is the minimization version of the goods EG primal.
This guess yields the following program which maximizes a convex function:
\begin{equation}
    \begin{aligned}
        \sup_{\beta \geq 0, p \geq 0} \quad & \sum_j p_j - \sum_i B_i \log(\beta_i) \\
        \text{ s.t.} \quad & \; p_j \leq \beta_i d_{ij} \quad\quad  \forall\; i,j.
    \end{aligned}
    \tag{Chores Dual}
    \label{chores EG dual}
\end{equation}

Henceforth we shall refer to this as the \emph{chores dual}, even though it is not a dual program in the formal sense of convex duality.
It is worth noting that, for example, we do not have global weak duality: some feasible pairs $(\beta,p)$ of \eqref{chores EG dual} yield a dual objective that does not lower bound the primal objective (this is easy to see, since the dual objective tends to infinity near poles).

Even though it is not a formal dual of the primal program, it turns out to behave in a dual-like manner at KKT points.
Indeed, the KKT points of this program yield chores CE, similar to how the primal EG for chores yields CE at all KKT points. 
We emphasize here that in much of the prior literature on Nash welfare minimization for chores, the focus has been on finding \emph{non-zero-disutility} KKT points; once we take the log of the Nash welfare objective, this constraint is implicitly enforced as a domain constraint on the log function, thus all KKT points become CE, both in the ``primal'' and ``dual'' EG programs for chores.

\begin{theorem}
    There is a one-to-one correspondence between competitive equilibria of a chores Fisher market and the KKT points of~\eqref{chores EG dual}.
    \label{thm:dual EG kkt correspondence}
\end{theorem}
\begin{proof}
    First, we show that any KKT point of~\eqref{chores EG dual} is a CE. 
    Let $(\tilde{p}, \tilde{\beta})$ denote a feasible point in the primal space of~\eqref{chores EG dual}, and $\tilde{x}$ be the dual variable associated with the constraint in~\eqref{chores EG dual}.
    Note that we overload notation slightly here: $(p, \beta, x)$ can, at this stage, only be understood to be primal and dual variables that arise from the programs, whereas we previously used these notations to refer to market quantities like prices and allocations. We will show that the primal and dual variables do in fact have these interpretations, thus justifying the notation.
    By the KKT conditions, any KKT point $(\tilde{p}, \tilde{\beta}, \tilde{x})$ of~\eqref{chores EG dual}  
    satisfies 
    \begin{enumerate}[(i)]
        \item $p, x \geq 0, \beta > 0$; 
        \label{dual-KKT-con-1}
        \item $d_{ij} \beta_i \geq p_j$ and $x_{ij} \left( d_{ij} \beta_i - p_j \right) = 0$ for all $i \in [n], j \in [m]$; \label{dual-KKT-con-2}
        \item $\sum_i x_{ij} \geq 1$ and $p_j \left( \sum_i x_{ij} - 1 \right) = 0$ for all $j \in [m]$; \label{dual-KKT-con-3}
        \item $\sum_j d_{ij} x_{ij} \geq \frac{B_i}{\beta_i}$ and $\beta_i \left( \sum_j d_{ij} x_{ij} - \frac{B_i}{\beta_i} \right) = 0$ for all $i \in [n]$.\label{dual-KKT-con-4}
    \end{enumerate}
    Note the $\beta > 0$ (instead of $\beta \ge 0$) in $(i)$ above. 
    This is implied by the domain constraint of the log function on $\beta_i$s. 
    By (\cref{dual-KKT-con-2}), we have $\frac{\tilde{p}_j}{d_{ij}} = \tilde{\beta}_i \geq \frac{\tilde{p}_{j'}}{d_{ij'}}$ for all $j' \in [m]$, if $\tilde{x}_{ij} > 0$ for all $i \in [n], j \in [m]$, which is equivalent to $(E2)$ in~\cref{Def:CE}. Combining (\cref{dual-KKT-con-2}) and (\cref{dual-KKT-con-4}), we have $\sum_j \tilde{p}_j \tilde{x}_{ij} = \sum_j d_{ij} \tilde{\beta}_i \tilde{x}_{ij} = B_i$ for all $i \in [n]$. Also, (\cref{dual-KKT-con-3}) leads to $\sum_i \tilde{x}_{ij} = 1$ if $\tilde{p}_j > 0$ for all $j \in [m]$. Thus, by definition, $(\tilde{p}, \tilde{x})$ constitutes a CE if we regard $p$ as price and $x$ as allocation.

    Next, we prove that any CE is a KKT point of~\eqref{chores EG dual}. Given any CE $(p, x)$, we denote $\beta_i := \frac{B_i}{\sum_j d_{ij} x_{ij}}$ for all $i \in [n]$. Then, $(p, \beta, x)$ trivially satisfies (\cref{dual-KKT-con-1}), (\cref{dual-KKT-con-3}) and (\cref{dual-KKT-con-4}). Since, for all $i$, $\sum_j d_{ij} x_{ij} = \max_{x'_i: \inp{p}{x'_i} \geq B_i} \sum_j d_{ij} x'_{ij}$ and $\sum_j d_ij x_ij \le \frac{d_{ij'}}{p_j} \cdot B_i$ for all $j'$, we have $d_{ij} \beta_i = \frac{d_{ij} B_i}{\sum_j d_{ij} x_{ij}} \geq \frac{d_{ij} B_i}{d_{ij} ({B_i} / {p_j}) } = p_j$ for all $i \in [n], j \in [m]$. For all $i \in [n]$, $\sum_{j \in [m]} x_{ij} \left( p_j - d_{ij} \beta_i \right) = B_i - B_i = 0$ since $\inp{p}{x_i} = B_i$. Then, (\cref{dual-KKT-con-2}) follows since $x_{ij} \left( p_j - d_{ij} \beta_i \right) \leq 0$ for all $i \in [n], j \in [m]$. 
    Hence, $(p, \beta, x)$ satisfies KKT conditions. 
\end{proof}

Next we look at another property typically satisfied by primal and dual convex programs: strong duality implies that optimal primal and dual solutions have the same objective. As we already noted, our \eqref{chores EG dual} does not even satisfy weak duality. Yet, it turns out that we do have a local form of strong duality, in the sense that the objectives are equal at KKT points. To see this, we next write more explicit versions of the primal and dual, in a way that keeps various important constants around for showing this equality (these same constants are needed in the goods setting).

Consider the primal and dual problems of chores EG in the following forms: 

\begin{equation}
    \begin{aligned}
    \inf_{x \geq 0, u \geq 0} \quad & f^p(x, u) := \sum_i B_i \log(u_i) \\ 
    \text{ s.t.} \quad & \inp{d_i}{x_i} \leq u_i \quad \hspace{20pt} \forall\; i \in [n] & \quad (\beta_i) \\ 
    & -\sum_i x_{ij} \leq -1 \quad\; \forall\; j \in [m] & \hspace{20pt} (p_j)
    \end{aligned} 
    \label{chores EG primal} 
\end{equation}
and 
\begin{equation}
    \begin{aligned}
    \sup_{\beta \geq 0, p\geq 0} \quad & f^d(\beta, p) := \sum_j p_j - \sum_i B_i \log{\beta_i} + \sum_i \big( B_i \log{B_i} - B_i \big) \\
    \text{ s.t.} \quad & p_j \leq d_{ij} \beta_i \quad\quad\quad \forall\; i \in [n],\, j \in [m] \hspace{30pt} (x_{ij}).
    \end{aligned} 
    \label{chores EG dual ext.} 
\end{equation}


\begin{theorem}
    If $(x, u, \beta, p)$ is a KKT point\footnote{Any pair of primal and dual variables which satisfies the KKT conditions is called a KKT point.} of~\eqref{chores EG primal}, then $(\beta, p, x)$ is a KKT point of~\eqref{chores EG dual ext.}.
    Reversely, if $(\beta, p, x)$ is a KKT point of~\eqref{chores EG dual ext.} and $u_i = \inp{d_i}{x_i}$ for all $i \in [n]$, then $(x, u, \beta, p)$ is a KKT point of~\eqref{chores EG primal}. 
    Furthermore, if $(x, u, \beta, p)$ is a KKT point of~\eqref{chores EG primal}, then 
    $$f^p(x, u) = f^d(\beta, p).$$ 
\end{theorem} 


\begin{proof}
    The first part is obvious due to the one-to-one correspondence between KKT points and CE for both programs. 
    Since $(x, u, \beta, p)$ is a KKT point of~\eqref{chores EG primal}, we have $\beta_i = \frac{B_i}{u_i}$ for all $i \in [n]$. 
    Since $(p, \beta, x)$ is a KKT point of~\eqref{chores EG dual ext.}, and thus a CE of the corresponding chores Fisher market, we have $\sum_j p_j = \sum_j \sum_i p_j x_{ij} = \sum_i \sum_j p_j x_{ij} = \sum_i B_i$. By these, we can obtain that 
    \begin{equation*}
        \sum_j p_j - \sum_i B_i \log{\beta_i} + \sum_i \big( B_i \log{B_i} - B_i \big) 
        = \sum_i B_i \log{\frac{B_i}{\beta_i}} + \sum_j p_j - \sum_i B_i 
        = \sum_i B_i \log{u_i}. 
    \end{equation*}
\end{proof}

Similarly to the primal EG for chores, the dual EG for chores has problems from an optimization perspective. 
One might try to find a KKT point by finding a local maximum of \eqref{chores EG dual} using some sort of iterative optimization method on the interior of the convex feasible region. However, this runs into trouble because \eqref{chores EG dual} has two ways that the objective tends to infinity, generally because we can either let $\beta_i$ tend to zero in order for $-\log \beta_i$ to tend to infinity (we refer to these as poles, as in the primal case), or by letting the prices tend to infinity, since their contribution to the objective grows faster than $-\log\beta_i$. 

In~\cref{fig:unbound-objective-directions}, we show a diagram which illustrates how naive gradient descent methods fail because they are either attracted to poles, or move along an unbounded direction. 
This instance is a Fisher market with $2$ agents and $1$ chore. In this instance, $B_1 = B_2 = 1$, $d_{11} = 2$ and $d_{21} = 1$. 
Obviously, there is only one KKT point / CE in this instance: $p_1 = 2$; $\beta_1 = 1, \beta_2 = 2$.
As we can see, even starting from some feasible points that are very close to the KKT point, trajectories generated by the gradient ascending algorithm are attracted by either of those poles.

\begin{figure}[t]
    \centering
    \subfloat[]{\includegraphics[width=0.42\linewidth]{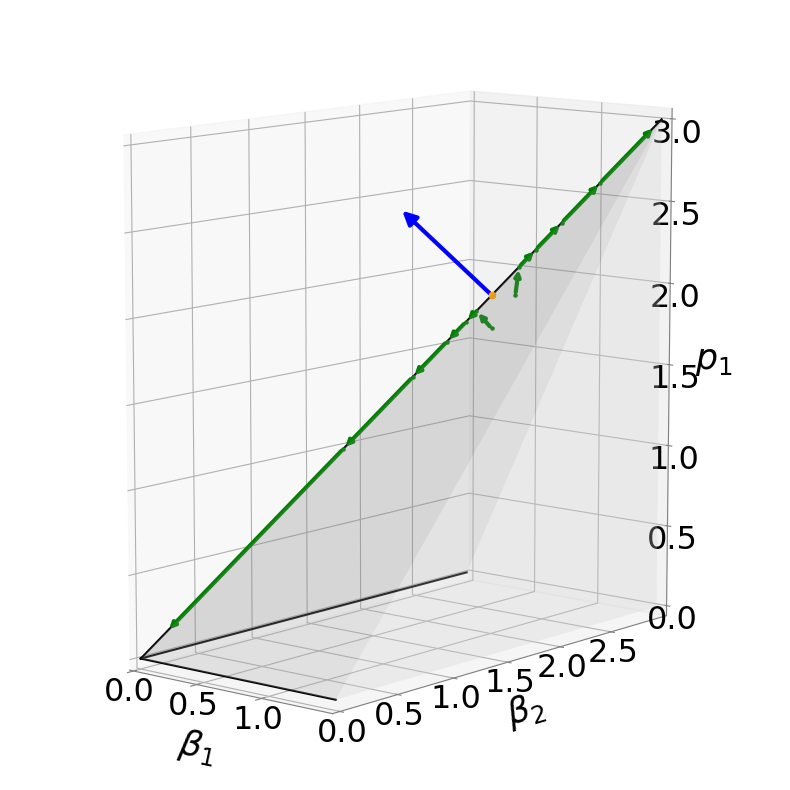}
    \label{fig:unbound-objective-directions}}
    \subfloat[]{\includegraphics[width=0.42\linewidth]{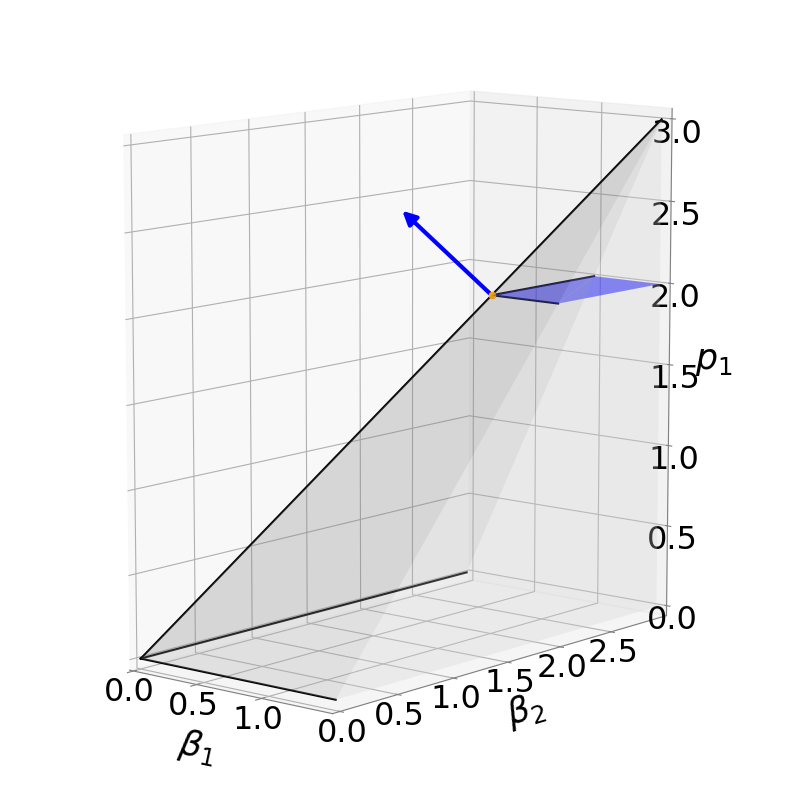}
    \label{fig:redundant-constraint}}
    \caption{An instance of a Fisher market with $2$ agents and $1$ chore. The gray region (polyhedral) is the feasible region of~\eqref{chores EG dual}. The (only) KKT point / CE ($(\beta_1, \beta_2, p_1) = (1, 2, 2)$) is marked by the yellow point. The blue arrow denotes the direction of gradient at the KKT point. In the left, two green trajectories are iterates generated by the gradient ascending algorithm, starting from two feasible points. In the right, the blue region (2-dimensional polyhedral) is the feasible region of~\eqref{chores dual redundant}.}
    \label{fig:unbound-objective-directions-and-redundant-constraint}
\end{figure}

Next, we show what we found to be a highly surprising fact: the \eqref{chores EG dual} can be ``fixed'' such that both poles and unbounded directions are avoided, without removing any KKT points.
This is achieved by adding a new constraint, which is redundant from a CE perspective. It is easy to see that in any CE, we must have that every buyer earns their budget exactly (otherwise they could do fewer chores and still satisfy their earning constraint), and since every item is allocated exactly, we must have that $\sum_j p_j = \sum_i B_i$.
Adding this constraint to \eqref{chores EG dual} yields the following program:
\begin{equation}
    \begin{aligned}
        \max_{\beta \geq 0, p \geq 0} \quad &\sum_j p_j - \sum_i B_i \log(\beta_i) \\
        \text{ s.t. } \quad & \quad\;\; p_j \leq \beta_i d_{ij} \quad \forall i,j\\
        & \sum_j p_j = \sum_i B_i
    \end{aligned} 
    \tag{Chores Dual Redundant}
    \label{chores dual redundant}
\end{equation}

The new constraint $\sum_j p_j = \sum_i B_i$ in \eqref{chores dual redundant} cuts off directions with infinite objective: 1) prices are upper bounded by $\sum_i B_i$, so the rays where prices go to infinity are avoided, 2) each $\beta_i$ is lower bounded by $\max_j p_j / d_{ij} \ge \sum_i B_i/ (m \cdot \max_j d_{ij})$ since $p_j \ge \sum_i B_i /m$ for some chore $j$. The following lemma formalizes this fact; the proof follows from the above discussion.

\begin{lemma}
    \label{lem:CDR finite objective}
    The optimal value of~\eqref{chores dual redundant} is at most $$\Big(\sum_i B_i\Big) \Big(1-\log(\sum_i B_i) + \log(m \max_{i,j} d_{ij})\Big).$$
\end{lemma}


Crucially, it turns out that the new constraint $\sum_j p_j = \sum_i B_i$ is redundant, in the sense that its Lagrange multiplier is zero at every KKT point. 
The redundancy also means that this constraint is implied by the KKT conditions of~\eqref{chores EG dual}, 
which implies that the one-to-one correspondence between CE and KKT points still holds for \eqref{chores dual redundant}.
\begin{theorem}
    There is a one-to-one correspondence between CE of the chores Fisher market and the KKT points of \eqref{chores dual redundant}.
    \label{thm:dual EG redundant kkt correspondence}
\end{theorem}
\begin{proof}
    Here, we refer to the two constraints in~\eqref{chores dual redundant} as the optimal disutility constraints (the first constraints) and the redundant constraint (the second constraint). 
    Let $x$ and $\mu$ be the dual variables corresponding to the optimal disutility constraint and redundant constraint, respectively. 
    By the KKT conditions, $(p, \beta, x, \mu)$ is a KKT point of~\eqref{chores dual redundant} if and only if it satisfies the conditions (\cref{dual-KKT-con-1}), (\cref{dual-KKT-con-2}) and (\cref{dual-KKT-con-4}) in~\cref{thm:dual EG kkt correspondence} as well as 
    \begin{enumerate}[(i)]
        \item $\sum_j p_j = \sum_i B_i$; 
        \label{dual-redundant-KKT-con-2}
        \item $\sum_i x_{ij} + \mu \geq 1$ and $p_j \left( \sum_i x_{ij} + \mu - 1 \right) = 0$ for all $j \in [m]$. \label{dual-redundant-KKT-con-3}
    \end{enumerate}
    First, we show that any KKT point of~\eqref{chores dual redundant} is a CE. 
    To do this, we only need to show $\mu = 0$ if $(p, \beta, x, \mu)$ is a KKT point of~\eqref{chores dual redundant}, 
    since the conditions (\cref{dual-KKT-con-2}) and (\cref{dual-KKT-con-4}) guarantee that every agent chooses her optimal bundle.  

    

    Note that  
    \begin{equation}
        \sum_j p_j (1 - \sum_i x_{ij}) = \sum_i B_i - \sum_j \sum_i p_j x_{ij} = 0, 
    \end{equation} 
    since $\sum_j p_j = \sum_i B_i$ by the redundant constraint and $\sum_j p_j x_{ij} = B_i$ for all $i$ by the optimal bundle condition. By KKT condition (\cref{dual-redundant-KKT-con-3}), we have
    \begin{equation}
        0 = \sum_{j \in [m]} p_j (1 - \sum_i x_{ij} - \mu) = \sum_{j \in [m]} p_j (1 - \sum_i x_{ij}) - \mu \sum_{j \in [m]} p_j = - \mu \sum_{i \in [n]} B_i. 
    \end{equation}
    This equality holds if and only if $\mu=0$. 

    Next, we prove that any CE is a KKT point of~\eqref{chores dual redundant}. Set $\mu=0$ and the rest follows similarly to the second part of the proof of~\cref{thm:dual EG kkt correspondence}. 
\end{proof} 

In~\cref{fig:redundant-constraint} we show how the redundant constraint restricts the feasible set in the instance in \cref{fig:unbound-objective-directions} in a way that avoids poles.



Because \eqref{chores dual redundant} does not have \infd s (\cref{lem:CDR finite objective}), it allows the use of first-order methods and interior-point methods on the convex polyhedron 
\begin{equation*}
    Y = \Big\{ (\beta, p) \in \RR^{n + m}_{\geq 0}:\, p_j \leq \beta_i d_{ij}, \;\forall\, i \in [n],j \in [m], \;  \sum_{j \in [m]} p_j = \sum_{i \in [n]} B_i \Big\}.
\end{equation*}
This is surprising because past work~\citep{BogomolnaiaMSY17,boodaghians2022polynomial} has repeatedly stressed that a challenge for computing CE for chores is that iterative methods that work in the interior are attracted to non-CE points. Our new program completely circumvents this issue. The remainder of the paper will be concerned with developing a provably fast and highly practical iterative method for computing (approximate) KKT points of \eqref{chores dual redundant}.
We note, though, that \eqref{chores dual redundant} opens the door to many other optimization methods, and we hope, for example, that efficient interior-point methods can be developed for it, as well as new first-order methods for large-scale problems.

\section{Greedy Frank-Wolfe Method}
\label{Sec:Alg}
In this section we introduce a simple first-order optimization method for computing chores CE based on \eqref{chores dual redundant}. The method only requires solving an LP at every iteration, and we will see that it has many desirable properties.

As above, let $Y$ be the convex polyhedron of the feasible set of~\eqref{chores dual redundant}, and let 
$y = (\beta, p)\in Y$ denote elements of $Y$.
Let $f$ be the \eqref{chores dual redundant} objective (after dropping $\sum_j p_j$, which sums to the constant $\sum_i B_i$), i.e., 
\begin{equation*}
    f(y) = -\sum_i B_i \log \beta_i. 
\end{equation*}

The algorithm we will use is a greedy variant of the Frank-Wolfe (FW) algorithm~\citep{frank1956algorithm}.
For a general differentiable minimization problem over a closed convex set, say $\min_{x \in \mathcal{X}} \tilde{f}(x)$, 
the FW method conducts the following steps iteratively. 
In every iteration $t$, the FW algorithm finds a minimizer $s_t$ to the linear approximation of the function $\tilde{f}$ around the previous iterate point $x_{t-1}$. Thereafter, the algorithm optimally chooses a step $\tau_t$ towards $s_t$ from $x_{t-1}$ and sets $x_{t} = x_{t-1} + \tau_t (s_t - x_{t-1})$ for $\tau_t \in [0, 1]$. 
The advantage of this method over other first-order methods like gradient descent is that it avoids projection (which can often be expensive). 


\begin{algorithm}[t]
    \begin{algorithmic}[1]
        \caption{Greedy Frank-Wolfe (GFW) algorithm}
        \Require{An initial feasible point in the dual space $y^0 = (\beta^0, p^0) \in \RR^{n + m}_{\geq 0}$ such that $\sum_j p^0_j = \sum_i B_i$}
        \For{$t = 1, 2, \ldots$}
            \State $y^{t+1} = \arg\max_{y\in Y}\, \inp{ \nabla f(y^t) }{ y - y^t }$
            \If{$\inp{ \nabla f(y^t) }{ y^{t + 1} - y^t } = 0$}
                \Return{$y^t$} 
            \EndIf
        \EndFor
        \label{algo:GFW}
    \end{algorithmic}
\end{algorithm}

For convex \emph{minimization}, the FW method has been immensely popular, for example in machine learning application~\citep{jaggi2013revisiting}, and it has previously been used to compute CE in various settings with goods~\citep{gao2020first,panageas2021combinatorial}.
In the goods setting, one must be careful about stepsize selection, in order to guarantee convergence; typically a step is taken from the current iterate in the direction of the point maximizing (or minimizing in the convex case) 
the linearization.
Interestingly, in the chores setting this is not the case; because we are \emph{maximizing} a convex objective, we can greedily move directly to the point returned by the linear maximization.
In particular, for our~\eqref{chores dual redundant} program, we can greedily choose our current iterate point to be the maximizer $s_t$ of the linear approximation of the function $f$ around $y_{t-1}$. That is, we set $y_t = s_t$ where $s_t = \argmax_{s \in Y} \inp{ \nabla f(y_{t-1}) }{ s - y_{t-1} }$. 
This yields~\cref{algo:GFW}. We call this algorithm the \emph{Greedy Frank Wolfe (GFW) algorithm}.

GFW has previously been studied for generic problems of maximizing a convex function in the context of certain machine learning applications~\citep{mangasarian1996machine,journee2010generalized}.
It is instructive to briefly understand \emph{why} we can be greedy with this method. Because we have a convex function, the following inequality holds,
\begin{align}
    f(y^{t+1}) - f(y^t) \geq \langle \nabla f(y^t), y^{t+1} - y^t \rangle.
    \label{eq:subgradient inequality}
\end{align}
Note that the right-hand side is exactly the objective in the line 2 of~\cref{algo:GFW}, and thus, since $y^t$ is a feasible solution, the RHS is weakly greater than zero. Moreover, the solution strictly improves, unless we have already found a local maximum (thus a KKT point and a CE).
This immediately implies that we have a monotonically-improving sequence of iterates. Secondly, note that there always exists a vertex of $Y$ which is optimal, and since there is a finite set of vertices in $Y$, this implies that the method converges in finite time.
These observations were formalized by \citet{mangasarian1996machine}:

\begin{theorem}[\citet{mangasarian1996machine}]
    Let $f$ be a differentiable convex function that is bounded above on a polyhedron $Y$. 
    Then,~\cref{algo:GFW} generates a sequence of iterates $\{ y^t \}_{t \geq 1}$ such that 
    \begin{itemize} 
        \item[(i)] $\{ y^t \}_{t \geq 1}$ is well-defined; 
        \item[(ii)] $f(y^t)$ strictly increases as $t$ increases; 
        \item[(iv)] The algorithm terminates with a $y^T$ satisfying the first-order optimality condition
        \begin{equation*}
            \inp{\nabla f(y^T)}{y - y^T} \leq 0 \quad \text{for all } y \in Y
        \end{equation*}
        in a finite number of iterations $T$.
    \end{itemize} 
\end{theorem}

\begin{figure}[t]
    \centering
    \includegraphics[width=0.6\linewidth]{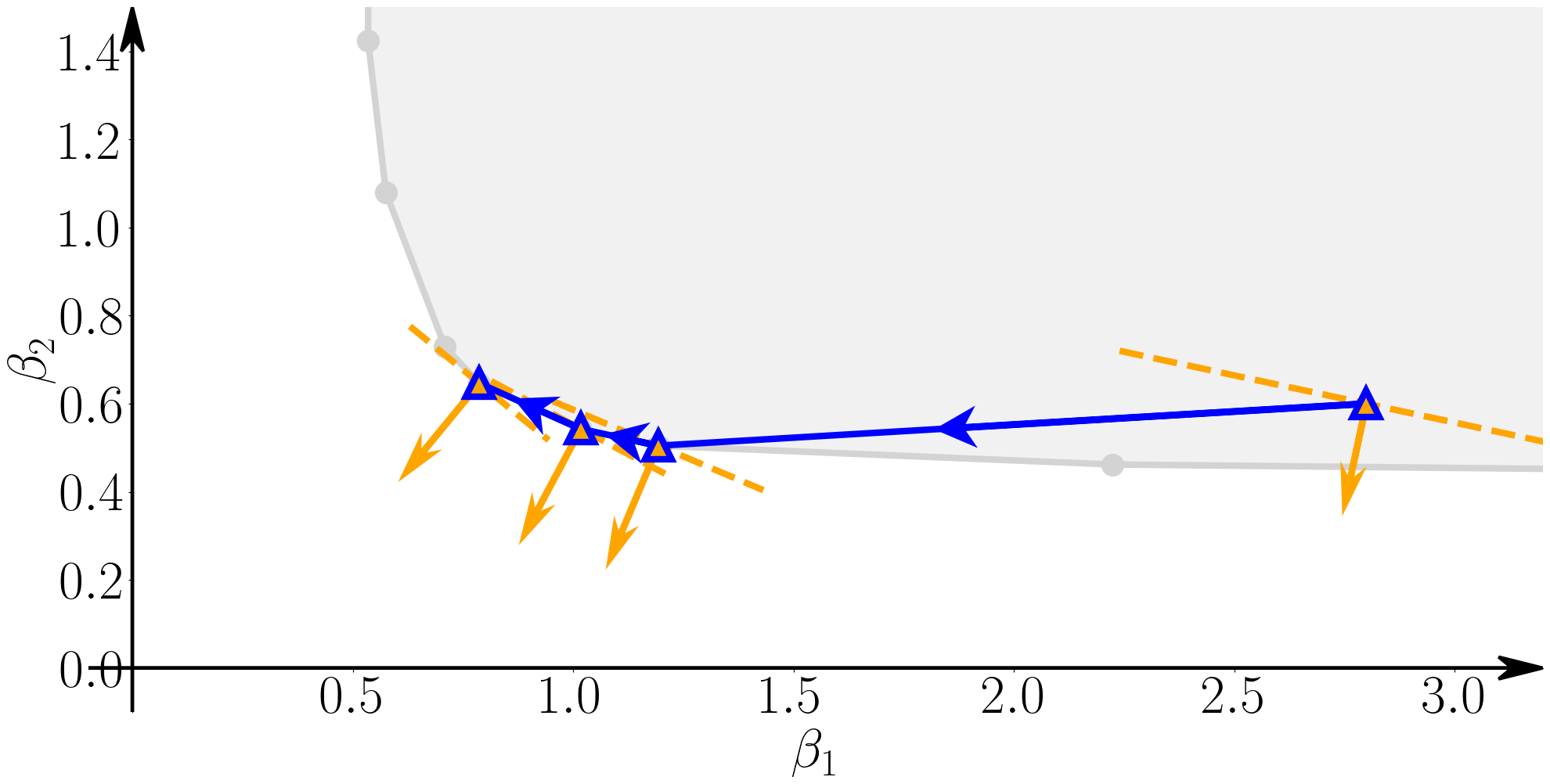}
    \caption{The trajectory of GFW until termination on a $2$-agent-$8$-chore instance, starting from a feasible interior point. The gray region is the feasible set projected onto the $(\beta_1, \beta_2)$ subspace. The blue lines with arrows denote the trajectory of GFW. The orange arrows and dotted lines denote gradient ascending directions of $f$ and corresponding planes perpendicular to these directions, respectively.}
    \label{fig:GFW}
\end{figure}

In~\cref{fig:GFW}, we illustrate a trajectory of the GFW algorithm on an instance with 2 agents and 8 chores. The earning requirements of agents are set to be one and disutilities are drawn from the [0, 1] uniform distribution. 

While the above guarantees finite-time convergence, it does not give a rate at which this convergence occurs. 
Later, some convergence-rate results were shown by \citet{journee2010generalized}. 
Most related to our chores CE setting, they show that for \emph{strongly convex} functions, GFW has a best-iterate convergence rate of $O(1/(\sigma_f \epsilon^2))$, where the convergence is in terms of some iterate $k \in \{1,\ldots,T\}$ having a small stepsize $\|y^{k+1}-y^k\| \leq \epsilon$, and $\sigma_f$ is the strong convexity parameter.
However, our objective is not strongly convex, and thus we cannot apply the results from \citet{journee2010generalized} to the chores CE setting.
In the next sections, we will show how to generalize the result of  \citet{journee2010generalized} to a class of \emph{Bregman-strongly-convex} functions. We then show that our objective indeed satisfies this generalized notion of strong convexity, and thus is guaranteed a small \emph{multiplicative} stepsize difference.
We go on to show that this implies convergence to chores CE.
In the appendix, we also give an alternative proof of this convergence directly from arguing about CE structure, without going through Bregman strong convexity.








\subsection{General GFW Convergence Results for Convex Maximization} 
\label{sec:general}


In this part, we generalize our Chores EG dual problem and show general convergence results for $\log$-type convex maximization over any nonempty closed convex set. 
We show that the multiplicative difference between two consecutive iterates is guaranteed to be smaller than any $\epsilon>0$ after $O(\frac{1}{\epsilon^2})$ number of iterations. 
We then show in the following section that such multiplicative ``stepsize convergence'' implies convergence to an approximate chores CE.

Formally, we consider the following $\log$-type convex maximization problem:
\begin{equation}
    \begin{aligned}
        \max_{x \in \R^n_{>0}} \quad & g(x) := -\sum_{i=1}^n a_i \log{x_i} \\ 
        \mbox{s.t.} \quad & x \in \mathcal{X},
    \end{aligned}
    \label{eq:log-type-convex-max-prob}
\end{equation}
where $a_i > 0, \; \forall\; i \in [n]$ and $\mathcal{X}$ represents any nonempty closed convex set. 
We only assume a finite upper bound for $g(x)$, i.e., there exists a $g^* < +\infty$ such that $g(x) \leq g^*$ for all $x \in \mathcal{X}$. 

To do this, we first generalize the results from~\citet{journee2010generalized} on convergence for strongly convex objectives to a convergence result for objectives that are strongly convex relative to some strictly convex differentiable function $h$.
This result is quite general, and may be of interest outside of our CE setting.

\begin{theorem}
    Let $f: \mathcal X \rightarrow \mathbb{R}$ be a continuously differentiable function that is $\alpha$-strongly convex relative to $h$, where $\mathcal X$ is a nonempty closed convex set, and bounded above by $f^* < +\infty$ on $\mathcal{X}$. 
    If we maximize $f$ over $\mathcal X$ using the GFW algorithm starting from an arbitrary feasible point $x^0$, then for any $\epsilon > 0$, we can find a pair of consecutive iterates $(x^t, x^{t+1})$ such that $D_h(x^{t+1}\,\Vert\,x^t) \leq \epsilon$ after at most 
    $\left\lceil\frac{f^* - f(x^0)}{\alpha\epsilon}\right\rceil$ iterations. 
    \label{thm:stepsize-convergence-under-generalized-strong-convexity}
\end{theorem}

\begin{proof}
    We first prove that the GFW is well-defined if $f$ is convex and bounded above by $f^* < +\infty$. 
    Since we have for each iteration $t$  
    \begin{equation*}
        0 \leq \inp{\nabla f(x^t)}{x^{t+1} - x^t} \leq f(x^{t+1}) - f(x^t) \leq f^* - f(x^t) < +\infty,
    \end{equation*}
    $\mathcal{X}$ has to be bounded on the direction of $x^{t+1} - x^t$. Also, $x^{t+1}$ is attainable since $\mathcal{X}$ is closed. 

    Let $T = \left\lceil \frac{f^* - f(x^0)}{\alpha\epsilon}\right\rceil$.
    Since $x^{t+1}$ maximizes $\inp{\nabla f(x^t)}{x}$ over $\mathcal X$ and $x^t \in \mathcal X$, 
    \begin{equation*}
        \inp{\nabla f(x^t)}{x^{t+1} - x^t} \geq 0 \quad \forall t = 0, 1, \ldots, T.
        \label{eq:maximize-LP-property-t}
    \end{equation*}
    By $\alpha$-strong convexity with respect to $D_h(\cdot\,\Vert\,\cdot)$, we have 
    \begin{equation*}
        f(x^{t+1}) \geq f(x^t) + \inp{\nabla f(x^t)}{x^{t+1} - x^t} + \alpha D_h(x^{t+1}\,\Vert\,x^t) \quad \forall\; t = 0, 1, \ldots, T. 
        \label{eq:generalized-strong-convexity-t}
    \end{equation*}
    Summing~\cref{eq:generalized-strong-convexity-t} up over $t = 0, \ldots,T$ and applying~\cref{eq:maximize-LP-property-t}, we have 
    \begin{equation*}
        \alpha T \min_{0 \leq t \leq T} D_h(x^{t+1}\,\Vert\,x^t) \leq \alpha \sum_{t=0}^T D_h(x^{t+1}\,\Vert\,x^t) \leq f(x^{T+1}) - f(x^0) \leq f^* - f(x^0). 
    \end{equation*}
    Thus, there exists a $t \leq T$ such that $D_h(x^{t+1}\,\Vert\,x^t) \leq \epsilon$ because $\min_{0 \leq t \leq T} D_h(x^{t+1}\,\Vert\,x^t) \leq \frac{f^* - f(x^0)}{\alpha T} \leq \epsilon$. 
\end{proof}

Now we consider the $\log$-type convex maximization case in \cref{eq:log-type-convex-max-prob}.
Let the reference function be $h(x) = -\sum_i \log{x_i}$. The Bregman divergence associated with $h$ is
\begin{equation}
    D_{IS}(y\,\Vert\,x) = D_h(y\, \Vert\, x) = \sum_i \left( - \log{\frac{y_i}{x_i}} + \frac{y_i}{x_i} - 1 \right), 
\end{equation}
where $D_{IS}( \cdot \Vert\, \cdot)$ refers to the \emph{Itakura–Saito distance}. Then, for the $\log$-type function in~\cref{eq:log-type-convex-max-prob}, we trivially have 
\begin{equation}
    g(y) \geq g(x) + \inp{\nabla g(x)}{y - x} + \min_i a_i \; D_{IS}(y\| x),
\end{equation}
since $g(y) - g(x) - \inp{\nabla g(x)}{y - x}$ can be seen as a sum over component-wise Itakura-Saito distances, each weighted by $a_i$, and due to nonnegativity, taking the minimum over $a_i$ can only make each term smaller.

Now, we note that $D_{IS}(y\,\Vert\,x)$ is lower bounded by $\frac{1}{3}\sum_i \big( \frac{y_i}{x_i} - 1 \big)^2$ when $\frac{y_i}{x_i}$ is close to $1$ for each $i$ (or, equivalently, $D_{IS}(y,\Vert\,x)$ is close to $0$). In particular, we have 
\begin{equation}
    \frac{1}{3}\sum_i \left( \frac{y_i}{x_i} - 1 \right)^2 \leq D_{IS}(y\,\Vert\,x)
    \label{eq:quadratic-function-lower-bound}
\end{equation}
if $- \log{\frac{y_i}{x_i}} + \frac{y_i}{x_i} - 1  \leq c := \sqrt{3} - 1 - \frac{1}{2}\log{3}$ for all $i$.

With this, we obtain the following result. 

\begin{theorem} 
    Suppose we solve~\cref{eq:log-type-convex-max-prob} using the GFW algorithm starting from an arbitrary feasible point $x^0$. Assume that there exists a finite upper bound $g^* < +\infty$. Then, for any $\epsilon > 0$, we can find a pair of iterates $(x^t, x^{t + 1})$ such that $\sum_i \big( \frac{x^{t+1}_i}{x^t_i} - 1 \big)^2 \leq \epsilon^2$ after at most 
    $$\left\lceil\frac{3(g^* - g(x^0))}{\min_i a_i} \frac{1}{\epsilon^2} + \frac{g^* - g(x^0)}{c \min_i a_i}\right\rceil$$
    iterations, where $c = \sqrt{3} - 1 - \frac{1}{2}\log{3} > 0$. 
    \label{crl:stepsize-convergence-for-log-type-convex-max}
\end{theorem}

\begin{proof}
Let $T = \left\lceil\frac{3(g^* - g(x^0))}{\min_i a_i} \frac{1}{\epsilon^2} + \frac{g^* - g(x^0)}{c \min_i a_i}\right\rceil$. 
We prove this result for two separate cases: $\epsilon > \sqrt{3c}$ and $\epsilon \leq \sqrt{3c}$. 

If $\epsilon > \sqrt{3c}$, we can find a pair of $(x^t, x^{t + 1})$ such that $D_{IS}(x^{t+1}\|x^t) \leq c < \epsilon^2/3$ after at most $\left\lceil\frac{g^* - g(x^0)}{c \alpha} \right\rceil \leq T$ iterations by~\cref{thm:stepsize-convergence-under-generalized-strong-convexity}. 

Since $- \log{a} + a - 1 \geq 0$ for all $a > 0$, we have $- \log{\frac{x^{t+1}_i}{x^t_i}} + \frac{x^{t+1}_i}{x^t_i} - 1 \leq c$ for all $i \in [n]$. 
Hence, by~\cref{eq:quadratic-function-lower-bound},
$$\frac{1}{3}\sum_i \left( \frac{x^{t+1}_i}{x^t_i} - 1 \right)^2 \leq D_{IS}(x^{t+1}\,\Vert\,x^t) \leq \frac{\epsilon^2}{3}$$

If $\epsilon \leq \sqrt{3c}$, then  $D_{IS}(x^{t+1}\,\Vert\,x^t) \leq \frac{\epsilon^2}{3} \leq c$ after $\frac{3(g^* - g(x^0))}{\min_i a_i} \frac{1}{\epsilon^2} \leq T$ by~\cref{thm:stepsize-convergence-under-generalized-strong-convexity}. Again, we have $- \log{\frac{x^{t+1}_i}{x^t_i}} + \frac{x^{t+1}_i}{x^t_i} - 1 \leq c$ for all $i \in [n]$, and by~\cref{eq:quadratic-function-lower-bound}, $\frac{1}{3}\sum_i \left( \frac{x^{t+1}_i}{x^t_i} - 1 \right)^2 \leq D_{IS}(x^{t+1}\,\Vert\,x^t) \leq \frac{\epsilon^2}{3}$. 
\end{proof}

\begin{remark}
    The result from \citet{journee2010generalized} does not apply to our setting since our objective is not strongly convex. It is possible that the extrapolation idea from \citet{gao2020first} could be applied in our setting in order to get some form of strong convexity and apply the results from \citet{journee2010generalized}. However, even if this approach were to work, it would lead to significantly worse polynomial dependence on $n$.
\end{remark}

\subsection{Convergence to Approximate Competitive Equilibrium}

With~\cref{crl:stepsize-convergence-for-log-type-convex-max} in hand, we are ready to show polynomial-time convergence to approximate competitive equilibrium. We attain this by showing that an $\epsilon$-small multiplicative stepsize implies that we have reached an $\epsilon$-approximate competitive equilibrium. 

Before showing this implication, first we show that prices are always strictly positive.
\begin{claim}
    The prices at a given iteration $t$ satisfy $p^t_j > 0$ for all $j \in [m]$. 
    \label{claim:all-price-strictly-positive}
\end{claim} 
\begin{proof}
    We prove this by contradiction. Assume there exists a $j \in [m]$, such that $p^t_j = 0$. Then note that $p^t_j/d_{i j} < \beta^t_{i}$ for all $i \in [n]$ (this follows from the fact that every $\beta^t_i$ is strictly positive). 
    Due to the strict inequality, there must exist some (possibly small) $\gamma>0$ such that if we increase $p_j^t$ by $\gamma$, and decrease the price of every other good by $\gamma / (\sum_i B_i -1)$, then we still have a feasible set of prices $p'$, and $\max_{j \in [m]} p^t_j / d_{ij}$ strictly decreases for all $i \in [n]$.
    Let $\beta'_i = \max_{j \in [m]} p_j'/d_{ij}$. Note that $\beta'_i < \beta^t_i$ for all $i \in [n]$ and thus $\sum_{i \in [n]} B_i \cdot \beta'_i/\beta^{t-1}_i <\sum_{i\in [n]} B_i \cdot \beta^t_i/\beta^{t-1}_i$ which contradicts the optimality of $\beta^t$ when solving the GFW LP.
\end{proof}

\begin{lemma} 
    Let $\{(p^t, \beta^t)\}_{t \geq 0}$ be a sequence of iterates generated by the GFW algorithm. 
    Let $x^t$ be the optimal dual variable of the LP subproblem at the $t$-th iteration and 
    \begin{equation*} 
        \bar{x}^t = \frac{\sum_i B_i}{\sum_i B_i \frac{\beta^t_i}{\beta^{t-1}_i}} x^t. 
    \end{equation*} 
    If $\big\lvert \frac{\beta^t_i}{\beta^{t-1}_i} - 1 \big\rvert \leq \epsilon$ for all $i \in [n]$, then $(p^t, \bar{x}^t)$ form an $\epsilon$-strongly approximate CE. 
    \label{lem:eps-approxi} 
\end{lemma} 
\begin{proof}
    Given the $t$-th iterate $(p^{t-1}, \beta^{t-1})$, the next iterate of the GFW algorithm is generated by the solution to the following LP: 
    \begin{equation} 
        \begin{aligned}
            \min_{p, \beta \geq 0} \quad & \sum_i B_i \frac{\beta_i}{\beta^{t-1}_i} &  \\ 
            \textnormal{s.t.} \quad & \quad\;\; p_j \leq d_{ij} \beta_i, \quad \forall\, i, j, & [x_{ij}] \\ 
            & \sum_j p_j = \sum_i B_i. & [\mu] \\ 
        \end{aligned} 
        \label{pgm:t-LP}
    \end{equation} 
    
    Let 
    $x$ and $\mu$ denote the dual variables corresponding to the first and second constraints of~\cref{pgm:t-LP}, respectively. 
    We use $(p^t, \beta^t)$ and $(x^t, \mu^t)$ to denote the optimal primal and dual variables at the $t$-th iteration, respectively.  

    By KKT conditions and strong duality of LP, $(p, \beta)$ is an optimal solution to~\cref{pgm:t-LP} if and only if there exists a pair of $(x, \mu)$ such that $(p, \beta, x, \mu)$ satisfies 
    \begin{enumerate}[(i)]
        \item $p, \beta, x \geq 0$; \label{t-LP-KKT-con-1}
        \item $p_j \leq d_{ij} \beta_i$ and $x_{ij} \left( p_j - d_{ij} \beta_i \right) = 0$ for all $i \in [n], j \in [m]$; \label{t-LP-KKT-con-2}
        \item $\sum_j p_j = \sum_i B_i$ and $\mu \left( \sum_j p_j - \sum_i B_i \right) = 0$; \label{t-LP-KKT-con-3}
        \item $\sum_i x_{ij} + \mu \geq 0$ and $p_j \left( \sum_i x_{ij} + \mu \right) = 0$ for all $j \in [m]$; \label{t-LP-KKT-con-4}
        \item $\frac{B_i}{\beta^{t-1}_i} - \sum_j d_{ij} x_{ij} \geq 0$ and $\beta_i \big( \frac{B_i}{\beta^{t-1}_i} - \sum_j d_{ij} x_{ij} \big) = 0$ for all $i \in [n]$.\label{t-LP-KKT-con-5}
    \end{enumerate}

    Recall that $(p^t, \beta^t)$ is an optimal solution to~\cref{pgm:t-LP}.
    By condition~\pcref{t-LP-KKT-con-2}, we have $\frac{p^t_j}{d_{ij}} = \beta^t_i \geq \frac{p^t_{j'}}{d_{ij'}}$ for all $j' \in [m]$ if $x^t_{ij} > 0$ for all $i, j$. 
    This guarantees every buyer is allocated her optimal bundle subject to the amount of budget she earns. 
    
    Second, we have 
    \begin{equation}
        B_i \frac{\beta^t_i}{\beta^{t-1}_i} \overset{\pcref{t-LP-KKT-con-5}}{=} \sum_j \beta^t_i d_{ij} x^t_{ij} \overset{\pcref{t-LP-KKT-con-2}}{=} \sum_j p^t_j x^t_{ij} \quad \forall\; i \in [n], 
        \label{eq:approx-earn-budget}
    \end{equation}
    which leads to $(1 - \epsilon) B_i \leq \sum_j p_j x_{ij} \leq (1 + \epsilon) B_i \leq \frac{B_i}{1 - \epsilon}$ since $\big| \frac{\beta^t_i}{\beta^{t-1}_i} - 1 \big| \leq \epsilon$ for all $i \in [n]$. That is, every buyer approximately earns her budget.
    
    Next, by summing up the equality in condition~\pcref{t-LP-KKT-con-4}, we obtain 
    \begin{equation}
        \sum_i \sum_j p^t_j x^t_{ij} + \mu^t \sum_j p^t_j = 0. 
    \end{equation} 
    Combining this with~\cref{eq:approx-earn-budget} and condition~\pcref{t-LP-KKT-con-3}, we have $\mu^t = - {\sum_i B_i \frac{\beta_i}{\beta^{t-1}_i}}\Big/{\sum_i B_i}$ and 
    \begin{equation}
        \sum_i x^t_{ij} = -\mu = \frac{\sum_i B_i \frac{\beta_i}{\beta^{t-1}_i}}{\sum_i B_i} \quad \text{if } p^t_j > 0. 
    \end{equation} 
    By~\cref{claim:all-price-strictly-positive}, this holds for all $j \in [m]$. 
    It follows that 
    \begin{equation}
        \sum_i \bar{x}^t_i = \frac{\sum_i B_i}{\sum_i B_i \frac{\beta^t_i}{\beta^{t-1}_i}} \sum_i x^t_i = 1 \quad \mbox{for all } j \in [m]. 
    \end{equation} 
    This means $\bar{x}$ is a feasible allocation (exactly allocating all chores).

    Therefore, the lemma follows by the definition of $\epsilon$-strongly approximate CE. 
\end{proof}
Therefore, by~\cref{lem:eps-approxi,crl:stepsize-convergence-for-log-type-convex-max}, we can conclude that the GFW algorithm finds an $\epsilon$-approximate CEEI in $\left\lceil \frac{3(g^* - g(x^0))}{\min_i a_i} \frac{1}{\epsilon^2} + \frac{g^* - g(x^0)}{c \min_i a_i} \right\rceil$ iterations. 
Observe that $g^*-g(x^0)$ can be upper bounded by $$\sup_{p, \beta \in Y} -\sum_i B_i \log(\beta_i).$$ To this end, note that for each agent $i$, $1/\beta_i \leq \min_j d_{ij}/p_j$ (implied by the constraints of our program). Since $\sum_j p_j = \sum_i B_i$, this implies that there exists at least one chore $j$, such that $p_j \geq \sum_i B_i/m$, implying that $d_{ij}/p_j \leq mD/(\sum_i(B_i))$, where $D = \max_{i,j} d_{ij}$. Therefore, we have  
\begin{equation*}
    -\sum_i B_i \log(\beta_i) = \sum_iB_i \log(1/\beta_i) \leq \sum_i B_i \log(mD/ (\sum_iB_i)) \leq n \max_i B_i \log(mD/\min_i B_i). 
\end{equation*}
\begin{theorem}[Complexity of finding an approximate CE]
\label{thm:convergence-final}
The GFW algorithm finds an $\epsilon$-strongly approximate CE in 
$\mathcal{O}(n \frac{\max_i B_i}{\min_i B_i} \log(\frac{mD}{\min_i B_i}) / \epsilon^2)$ iterations.
\end{theorem} 


For CEEI, our results (\cref{thm:convergence-final}) improve the current state-of-the-art convergence rate in terms of the number of iterations relative to $\epsilon$: The EPM by \citet{boodaghians2022polynomial} requires $\mathcal{\tilde{O}}(n^3/\epsilon^2)$ iterations, and the combinatorial algorithm by~\citet{ChaudhuryGMM22}  requires $\mathcal{\tilde{O}}(nm/\epsilon^2)$ iterations, in contrast to $\mathcal{{\tilde{O}}}(n/\epsilon^2)$ iterations of GFW\footnote{For CEEI, note that $\max_iB_i = \min_i B_i = 1$.}, to reach an approximate CE. Furthermore, we expect the cost of each iteration to be cheaper for GFW. In GFW, we only need to solve an LP in every iteration. Moreover, the approximate solution of the LP poses no numerical issues. In contrast, solving the convex QP in EPM is a more challenging task. It is hard to find exact running times for general convex QPs, and even quite good approximate solutions can lead to serious numerical issues, as we show in \cref{sec:experiments}. Further, we believe that the LP solved in every iteration of GFW can be solved (in worst case) using algorithms faster than the ones used for general-purpose LPs~\cite {van2020deterministic}, and we leave this as an interesting question for future research.

The result in~\cref{thm:convergence-final} can also be recovered with mildly stronger guarantees through a primal-dual analysis of the GFW (See~\cref{sec:polytime-conv-approxi}).  The analysis in~\cref{sec:polytime-conv-approxi} gives us further insights on each run of GFW: for instance, in each iteration of GFW where we do not yet have a good approximate CE point, we make a large improvement on the objective function (\cref{progress}).

 

\section{Experiments} 
\label{sec:experiments}


In~\cref{Sec:Alg}, we showed that the GFW algorithm guarantees convergence to a strongly approximate CE in polynomial time, and converges to an exact CE in finite number of iterations, while requiring solving only one LP per iteration.
In contrast, the prior numerical state-of-the-art method, the EPM of \citet{boodaghians2022polynomial} requires solving a convex quadratic program (QP) at every iteration. 

In this section, we investigate the numerical performance of GFW as compared to the EPM, when both utilize state-of-the-art optimization software for solving LPs and QPs (in particular, using Gurobi). 
We remark that we do not provide comparisons with the combinatorial algorithm~\citep{ChaudhuryGMM22}. This is primarily attributed to the numerical issues that arise with the arithmetic performed by combinatorial algorithms--in particular, the bit-length of the prices and the allocation can grow exponentially with the iterations in the combinatorial algorithms, as conjectured by~\citet{DuanM15}.\footnote{In private communication, we are aware that Omar Darwish and Kurt Mehlhorn have an implementation of the combinatorial algorithm for Arrow-Debreu markets with goods, and they confirm the conjecture in~\citet{DuanM15}.}

In~\cref{algo:EPM} we give pseudocode for the EPM for computing an exact CE for chores. 
The set $\mathcal{D}^+$ in the pseudocode denotes the set of all possible disutility profiles corresponding to any feasible allocation or over-allocation, i.e., 
\begin{equation*}
    \mathcal{D}^+ = \Big\{ d \in \mathbb{R}^n_{\geq 0} \;\vert\; \exists\; x \in \mathbb{R}^{n \times m}_{\geq 0} \mbox{ such that } d_i \geq \sum\nolimits_j d_{ij} x_{ij} \;\forall\; i \in [n], \sum\nolimits_i x_{ij} = 1 \;\forall\; j \in [m] \Big\}. 
\end{equation*}
\begin{algorithm}[t]
\begin{algorithmic}[1]
    \caption{Exterior Point Method (EPM)}
    \Require{An initial infeasible disutility profile $d^0$ such that $d^0_i > 0$ and close to $0$$\quad \forall\;i \in [n]$}
    \For{$k = 1, 2, \ldots$}
        \If{$d^k$ is infeasible}
            \State $d^{k, *} = {\arg\min}_{d \in \mathcal{D}^+ \cap \{ d \vert d \geq d^k \} } \big\lVert d - d^k \big\rVert$ \footnotemark 
            \State $a^k \gets {(d^{k, *} - d^k) n}\big/{\inp{d^{k, *} - d^k}{d^{k, *}}}$
            \State $d^{k+1} \gets {1}\big/{a^k}$
        \Else
            \State \Return{$(d^k, a^{k-1})$} 
        \EndIf
    \EndFor
    \label{algo:EPM}
\end{algorithmic}
\end{algorithm}

\paragraph{Experimental Setup}
We compare the algorithms on a variety of instances. First, we consider 
randomly generated instances. For a given instance size $n$, we construct $n$ buyers, each with a budget of $1$.
Then, we generate $n$ chores. This means that disutilities are an $n\times n$ matrix $D$. Each cell $d_{ij}$ in $D$ is sampled i.i.d. from some distribution $F$: $d_{ij} \sim F$.
We consider five different distributions of disutilities:
(1) uniform on $[ 0,1)$;
(2) log-normal associated with the standard normal distribution $\mathcal{N}(0, 1)$; 
(3) truncated normal associated with $\mathcal{N}(0, 1)$ and truncated at ${10}^{-3}$ and $10$ standard deviations from $0$; 
(4) exponential with the scale parameter $1$; 
(5) uniform random integers on $\{1,\ldots,1000\}$.
For each distribution, we consider instance sizes $n\in \{2,50,100,150,200,250,300\}$.
For each choice of distribution and $n$, we generate 100 instances.

Secondly, we generate a set of chores instances by considering a dataset of bids from PC members at the AAMAS conference on potential papers to review, obtained from PrefLib~\citep{mattei2013preflib}. We use the 2021 dataset, which has 596 PC members included in the dataset.
Each PC member scored each paper ordinally as \{yes $>$ maybe $>$ no response $>$ no $>$ conflict\}.
In order to generate real-valued disutilities, we convert these into the values \{1, 3, 5, 7, $4000$\}.
We then generate two different datasets:
In the first, we use these 5 values directly, and sample subsets of $n$ PC members and $n$ papers as follows: first uniformly choose one paper at random, then choose a subset of papers (of size $n$) which are the most similar to it (considering the PC members' responses as feature vectors of the paper), finally choose the subset (of size $n$) of PC members that have the largest number of non-conflict responses on the chosen papers. 
In the second, we additionally add Gaussian noise to each valuation, using $\mathcal{N}(0, 0.04)$. 
For both datasets, we consider instance size $n \in \{ 2, 50, 100, 150, 200, 250, 300 \}$. For each dataset and $n$, we generate 100 instances.

We ran all experiments on a personal desktop that uses the Apple M1 Chip and has 8GB RAM. 
We used Gurobi version 11.0.0. 
Since the tolerance parameters for Gurobi models may influence the accuracy of the solutions per iteration and thus the performance of EPM, 
we tuned those parameters. 
We considered 4 parameters: we tuned 
\texttt{FeasibilityTol} in the range of $10^{-9}$ to $10^{-6}$, 
\texttt{OptimalityTol} in the range of $10^{-9}$ to $10^{-6}$, 
\texttt{BarConvTol} in the range of $0$ to $10^{-8}$, 
and \texttt{BarCorrectors} in the range of $0$ to $10^{4}$~\footnote{The default setting chooses this parameter automatically, depending on problem characteristics.}
.
For our final experiments, we set these parameters as \texttt{FeasibilityTol} = $10^{-6}$ (default), \texttt{OptimalityTol} = $10^{-6}$ (default), \texttt{BarConvTol} = $0$ (minimum), and \texttt{BarCorrectors} as default. This is because shrinking \texttt{BarConvTol} significantly increases the accuracy of solutions; but decreasing the other tolerance parameters, or increasing \texttt{BarCorrectors}, greatly increases the time Gurobi spends per iteration, thus harming the EPM performance, with no significant improvement to the solution accuracy. See \cref{sec:appendix experiments} for plots illustrating this point.

In our experiments, we 
measure the computational cost required to find a strongly-approximate CE and an exact CE. 
In particular, we measure the number of iterations and the running time. 
Recall the definitions of $\epsilon$-approximate CE and exact CE in~\cref{Sec:prelim}.
We consider any condition with $\epsilon\leq 10^{-6}$ as a condition that holds exactly. 
We consider any CE with the first condition relaxed to $\epsilon \leq 0.01$ (and the other two conditions exactly satisfied) as a strongly-approximate CE. 
To measure the equilibrium approximation quality $\epsilon$ at a given iteration, we compute a pair $(p, x)$ corresponding to the current iterate for each algorithm. 
For GFW, we compute $(p, x)$ based on~\cref{lem:eps-approxi}, where $x$ is given by the optimal dual variable of the corresponding LP. For EPM, we solve an LP to find a feasible $x$ corresponding to the final $d^k$ if the algorithm finds an exact CE, and retrieve $x$ corresponding to $d^{k, *}$ if $(d^{k, *}, a^k)$ corresponds to a strongly-approximate CE. We compute prices $p$ for EPM via $p_j = \min_i a_i d_{ij}$~\citep{boodaghians2022polynomial}. 

Though the EPM is supposed to reach an exact CE when it finds a feasible (or over-allocated) disutility profile, we found in our experiments that an unignorable numerical error can be caused by the approximate solution of QPs.
If the projection step in~\cref{algo:EPM} cannot be solved accurately, this can cause $a^{k-1}$, the normal vector of the generated hyperplane, to point in the wrong direction. 
This leads to an inaccurate price vector $p^*$. 
Moreover, let $d^k$ be the true iterate with a perfect-precision QP solution, and let $\tilde{d}^k$ be the ``solution'' found by Gurobi at a problematic iteration.
The Gurobi solution $\tilde{d}^k$ may end up being feasible even if $d^k$ is not, or vice versa. This may lead to incorrect termination decisions in either case (in which case our attempt to compute a corresponding CE allocation fails).
This is a critical issue, because we have no way to proceed once this happens.
Considering this issue, we say that the EPM fails if we cannot find a corresponding strongly approximate CE even though the EPM ostensibly terminated with one.
In our experiments, we also measure the fraction of instances successfully solved without this issue.

\paragraph{Experiment Results and Discussion}
Our results on the five different distributions for i.i.d. random disutilities are shown in \cref{fig:random chores}.
\begin{figure}
    \centering
    \makebox[0pt][r]{\makebox[30pt]{\raisebox{40pt}{\rotatebox[origin=c]{90}{uniform(0,1)}}}}%
    \includegraphics[width=0.28\linewidth]{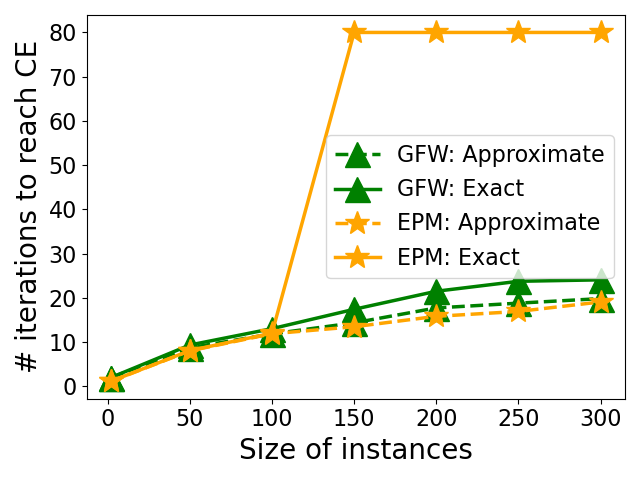}
    \includegraphics[width=0.28\linewidth]{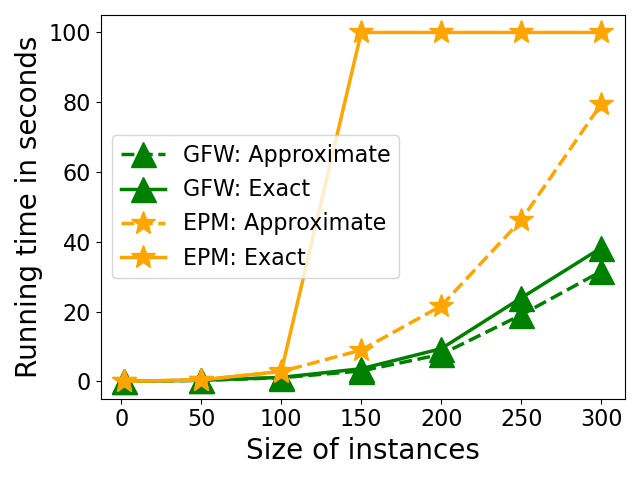}
    \includegraphics[width=0.28\linewidth]{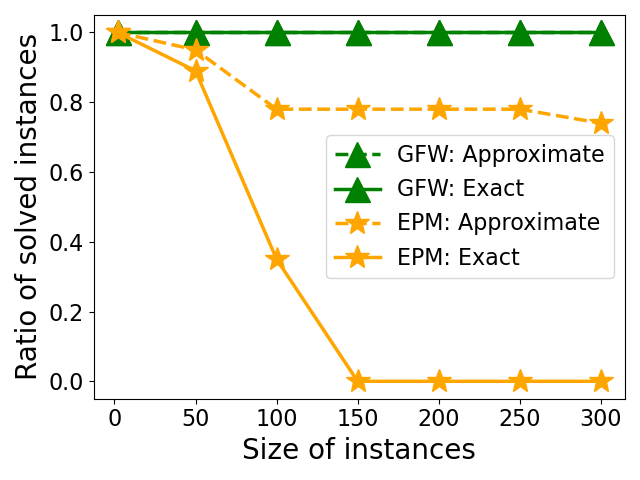}
    \makebox[0pt][r]{\makebox[30pt]{\raisebox{40pt}{\rotatebox[origin=c]{90}{lognormal}}}}%
    \includegraphics[width=0.28\linewidth]{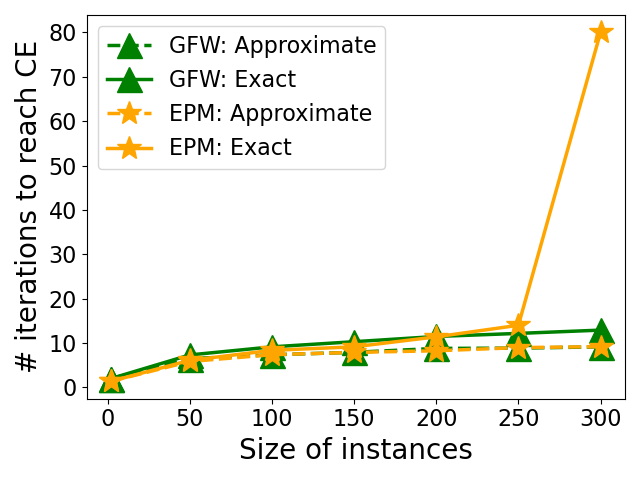}
    \includegraphics[width=0.28\linewidth]{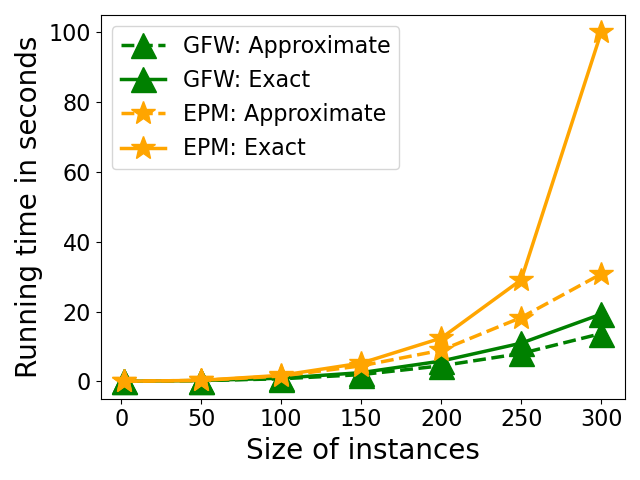}
    \includegraphics[width=0.28\linewidth]{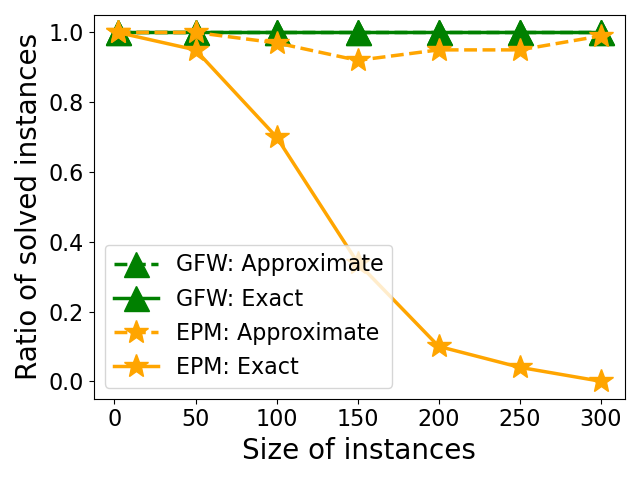}
    \makebox[0pt][r]{\makebox[30pt]{\raisebox{40pt}{\rotatebox[origin=c]{90}{truncated normal}}}}%
    \includegraphics[width=0.28\linewidth]{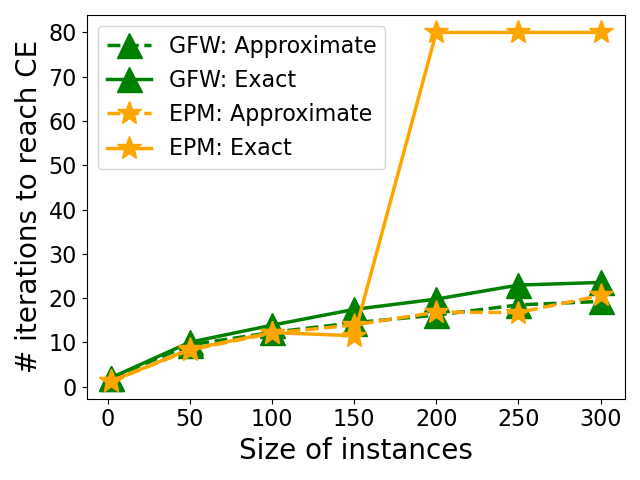}
    \includegraphics[width=0.28\linewidth]{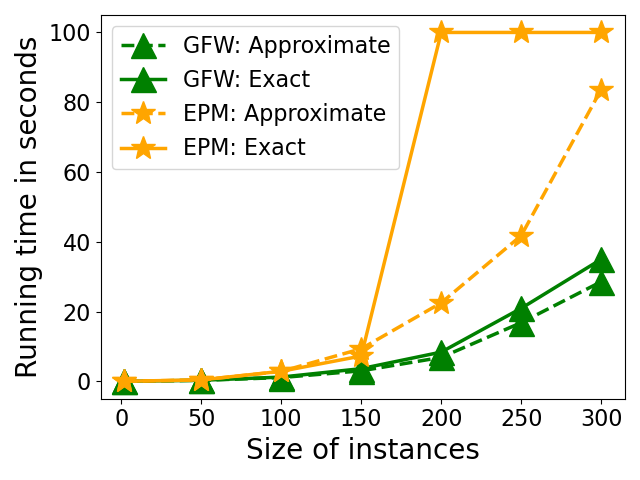}
    \includegraphics[width=0.28\linewidth]{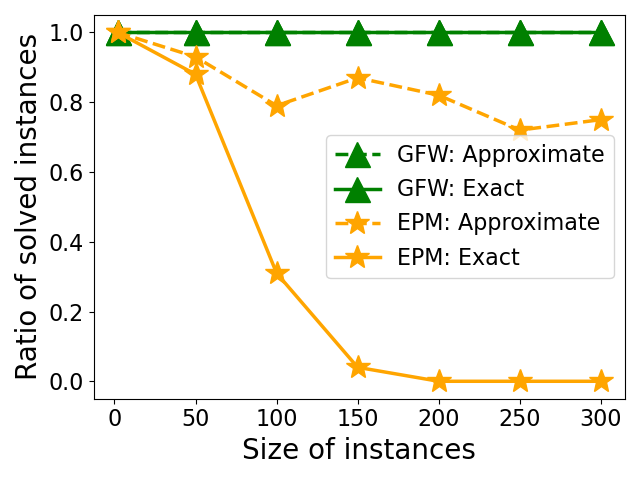}
    \makebox[0pt][r]{\makebox[30pt]{\raisebox{40pt}{\rotatebox[origin=c]{90}{exponential}}}}%
    \includegraphics[width=0.28\linewidth]{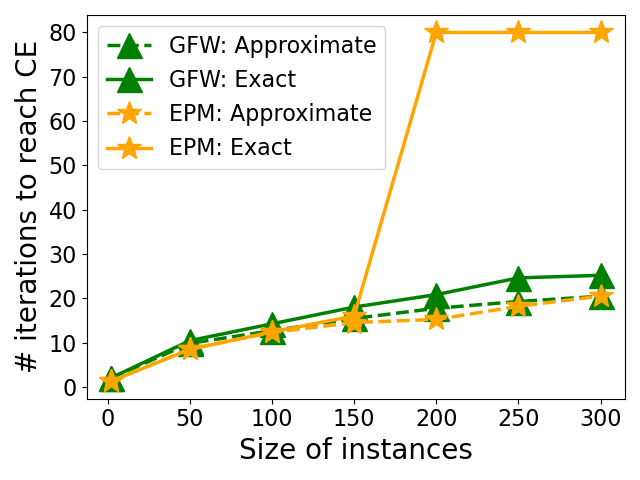}
    \includegraphics[width=0.28\linewidth]{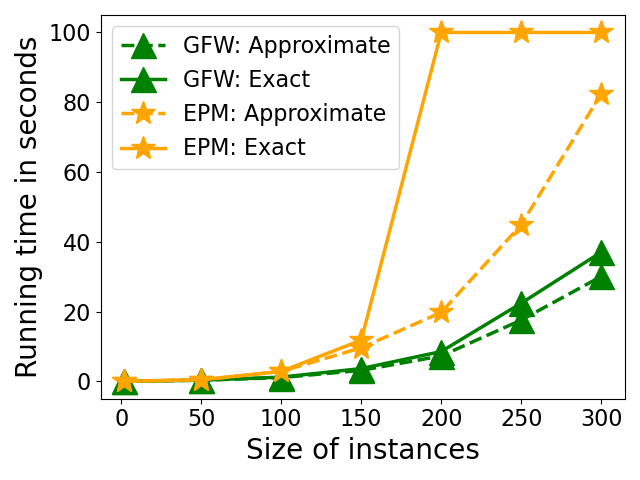}
    \includegraphics[width=0.28\linewidth]{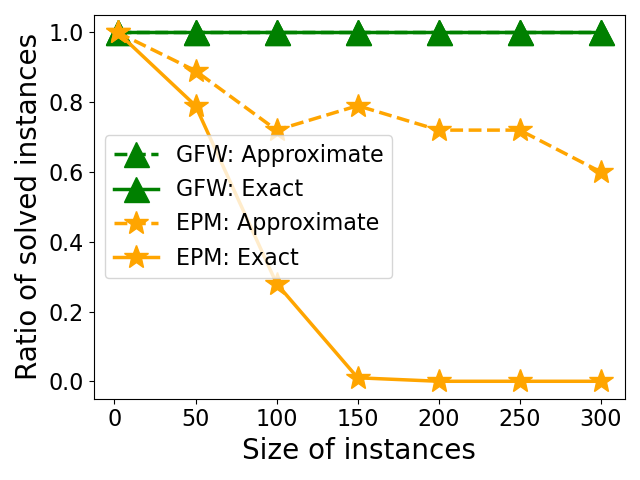}
    \makebox[0pt][r]{\makebox[30pt]{\raisebox{40pt}{\rotatebox[origin=c]{90}{randint(1,1000)}}}}%
    \includegraphics[width=0.28\linewidth]{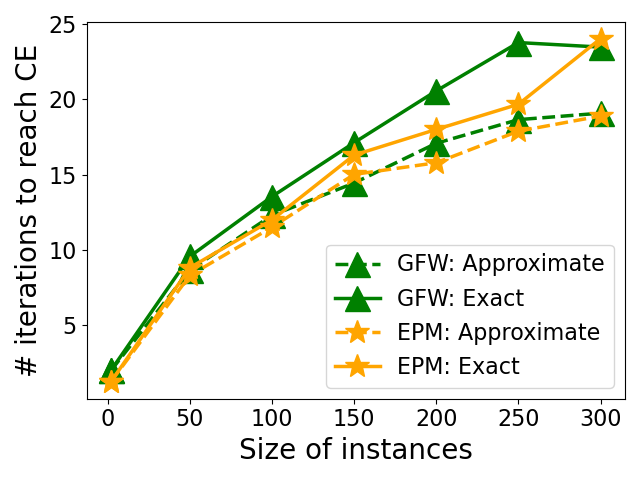}
    \includegraphics[width=0.28\linewidth]{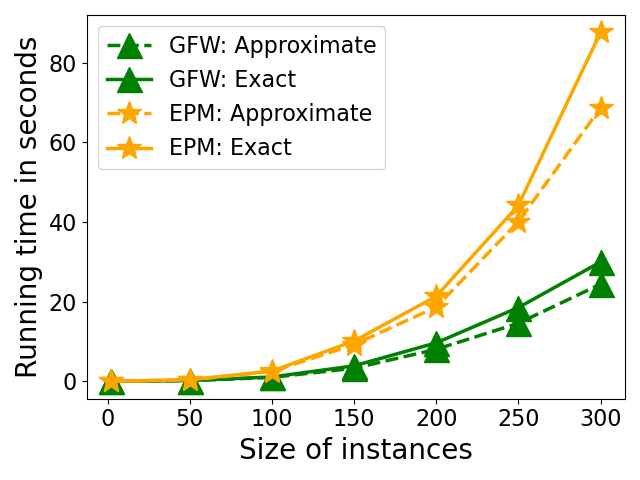}
    \includegraphics[width=0.28\linewidth]{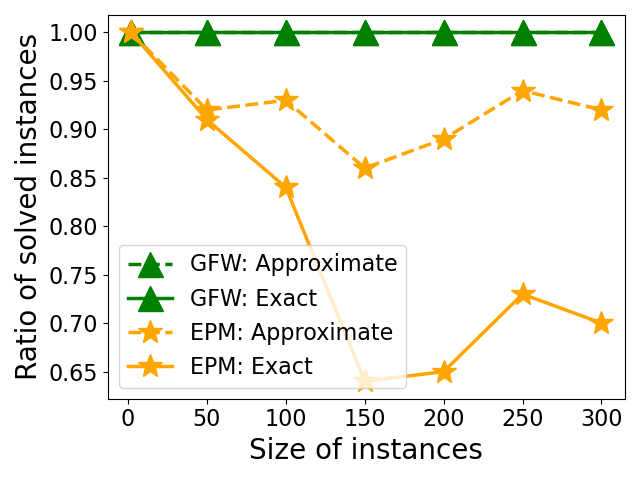}
    \caption{Numerical results on iid randomly generated disutility matrix instances. 
    Each row corresponds to a different disutility distribution (noted on the left).
    Left: Average number of iterations to solve a given instance size.
    Middle: Average wall clock running time to solve a given instance size (using Gurobi for both algorithms).  
    Right: Fraction of instances solved for each algorithm. 
    GFW solves every instance for every distribution. 
    }
    \label{fig:random chores}
\end{figure}
In the left column of plots, we show the average number of iterations taken for each algorithm before finding a solution.
In the middle column of plots, we show the mean running time.
For both the left and middle column of plots, we only average over instances that were solved by EPM (and if EPM reaches the largest number on the y axis it denotes that  no instances were solved at a given size).
The right column shows the fraction of instances solved for each instance size.

In terms of the number of iterations taken by each algorithm, we see that that they are very comparable: in all cases where both algorithms solve some instances, the number of iterations is quite low: always below 30 iterations. When it does work, the EPM seems to take slightly fewer iterations than GFW in some cases.
In terms of runtime, GFW is generally faster than EPM: in every case where EPM solves at least one instance, we have the GFW is faster, sometimes by a factor of 2-4.
By far the most important result in terms of practicality is the right column of plots. We see that GFW successfully solves every instance we tried, out of the 3500 instances generated in total.
In contrast, EPM fails on a significant number of instances, sometimes failing on every instance for larger instance sizes, when attempting to compute exact CE.

The results on our semi-synthetic AAMAS bidding instances are shown in \cref{fig:bidding-data-sampled}.
Again we see that the number of iterations to reach CE is quite comparable across GFW and EPM, on the instances where EPM succeeds. The running time results are also similar: both for exact and approximate CE, GFW is faster than EPM in both the original data setting, and the added-noise setting. Notably, both algorithms are very fast, terminating in less than ten seconds for the original data, and GFW still terminates in less than 15 seconds for the added-noise setting.
On the original data, both algorithms solve every instance. This is likely due to the fact that the original data is a very special type of instance: there are only 5 possible distinct disutility values, and every buyer has a lot of chores with value 1. Thus, the equilibrium is very simple, which is also evident in the fact that the mean number of iterations is 3.5 even for the largest instances.
On the added-noise data, we see again that EPM exhibits significant numerical issues, failing on a large fraction of instances, especially for exact CE.

Overall, we conclude that GFW is the first highly practical algorithm: it solves every instance in very little time and in few iterations, even for instances where the size of the disutility matrix is $300\times 300$. In comparison, EPM fails to solve many instances in our datasets.
As discussed before, we do caution that the success of EPM is highly dependent on the performance of the Gurobi QP solver, and in particular its accuracy. We found that the QP performance of Gurobi can vary substantially across different machines, and unpredictably so.
Due to this so-called ``performance variability'', it is difficult to conclusively evaluate the wall-clock performance of the algorithms, as both algorithms rely on the performance of Gurobi on large-size instances. 
In particular, some preliminary experiments suggests that the wall-clock performance of EPM improves relative to the wall-clock performance of GFW if one adds more RAM.
However, we found that even in this case, EPM still has a significant failure rate, since Gurobi still outputs approximate QP solutions, and this still leads to numerical errors that break EPM.

\begin{figure}
    \centering
    \makebox[0pt][r]{\makebox[10pt]{\raisebox{40pt}{\rotatebox[origin=c]{90}{original}}}}%
    \includegraphics[width=0.85\linewidth]{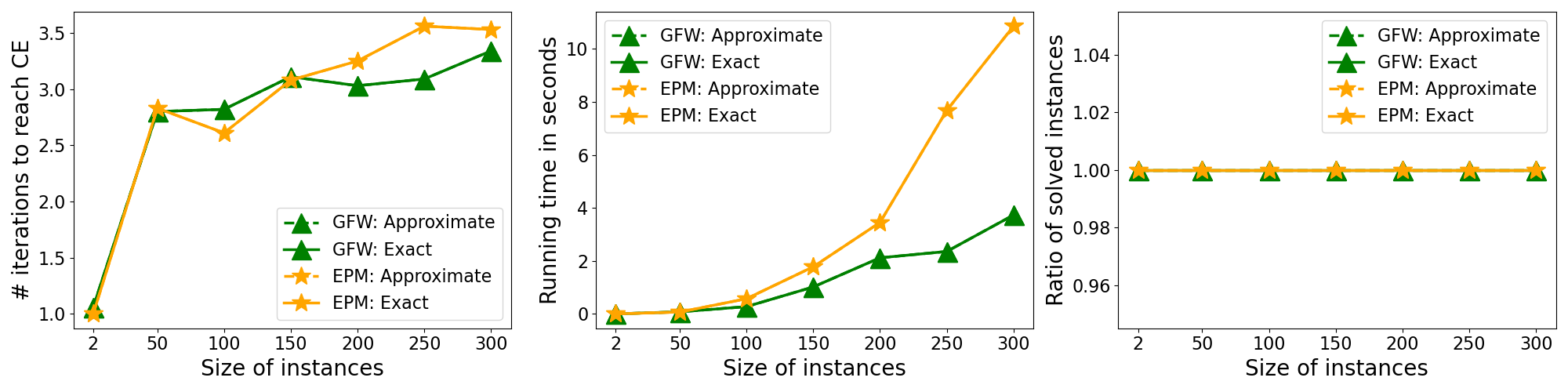}
    \makebox[0pt][r]{\makebox[10pt]{\raisebox{40pt}{\rotatebox[origin=c]{90}{with noise}}}}%
    \includegraphics[width=0.85\linewidth]{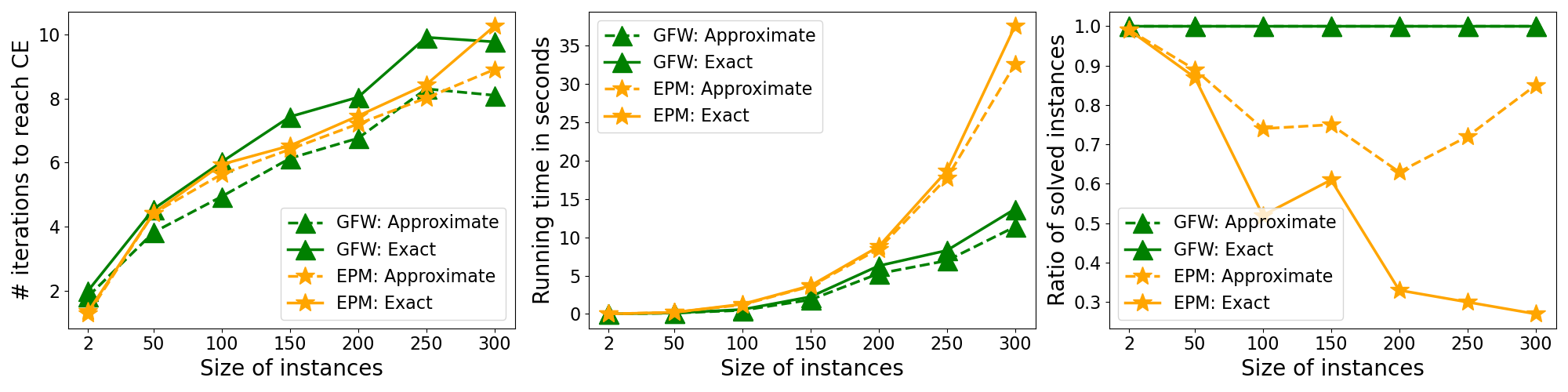}
    \caption{Numerical results on AAMAS PC member bidding data. 
    The top row uses the ``original'' valuations.
    The bottom row adds Gaussian noise.
    Left: Average number of iterations to solve a given instance size.
    Middle:  Average wall clock running time to solve a given instance size.
    Right: Fraction of instances solved.
    }
    \label{fig:bidding-data-sampled}
\end{figure}





\section{Conclusion and Discussion}
\label{sec:conclusion-and-discussion}
We develop a new convex maximization program for computing CE in Fisher markets with chores, and showed that it satisfies some dual-like properties with the concave minimization program for computing CE.
We then derived a new redundant constraint which allows us to avoid directions that lead to infinitely positive objective. This yielded the first mathematical program for finding CE for chores such that it is possible to run iterative optimization methods over a convex feasible region.
We then introduced the GFW method for computing CE. From prior results on convex maximization over polyhedra, it follows that this method terminates in a finite number of iterations. We then showed that, in fact, the method finds an $\epsilon$-approximate CE in $O(\frac{1}{\epsilon^2})$ iterations, while requiring solving only a simple LP at every iteration. We also gave a more general convergence result that generalizes prior results for strongly convex functions to Bregman-strongly-convex functions.
Finally, we showed through extensive numerical experiments that our method is highly practical: it solved every instance that we tried in a minute or less and is both robust and simple to implement. In contrast, we showed that the prior state-of-the-art method, the EPM, often fails to solve instances altogether, due to numerical accuracy problems that break the algorithm. 

\paragraph{Challenge of moving from approximate CE to exact CE.} The obvious open problem is to settle the exact complexity of finding a CE in the chores market. A natural attempt would be to find a good approximate CE  and then argue how to move to a ``nearby'' exact CE. \emph{Unfortunately, such approaches can face the issue that there are instances where an approximate CE can be very far from an exact CE (See~\cref{example-CE} for an example)}. 

Therefore, one may need to search for additional properties while looking out for a good approximate CE/ CEEI. Note that our GFW method indeed finds an exact CE in the chores market in finite time. In fact, it finds an exact CE in very few iterations in our experiments. Thus, it would be interesting to investigate instances where the time taken by GFW to reach an exact CE is exponential. We believe that this can improve our understanding of ``hard instances'' (if any) for iterative methods for computing a CE.

Despite the above hurdle, we believe that finding a $\epsilon$-approximate CE in the chores market with a better time complexity dependence on $\frac{1}{\epsilon}$ is of theoretical interest. An ideal goal would be to have a poly-logarithmic dependence on $\frac{1}{\epsilon}$, but any polynomial improvement would be a stepping stone (currently, all methods have a dependence of $\mathcal{O}(\frac{1}{\epsilon^2})$). 

\vspace{-12pt}

\section*{Acknowledgement}
Christian Kroer was supported by the Office of Naval Research awards N00014-22-1-2530 and N00014-23-1-2374, and the National Science Foundation awards IIS-2147361 and IIS-2238960. Ruta Mehta was supported by the NSF grant CCF-2334461.


\vspace{-12pt}

\bibliographystyle{informs2014} 
\bibliography{refs}

\appendix
\section{A Primal-Dual Analysis of the Convergence of the Frank Wolfe Method to an Approximate Equilibrium}
\label{sec:primal-dual-convergence}
Since $f$ does not depend on $p$, we ease notation in this by section by writing $f(\beta)$ for the \eqref{chores EG dual} objective function, which was previously defined as depending on both $p$ and $\beta$. 

Let $\{(p^t, \beta^t)\}_{t \geq 0}$ be the sequence of iterates generated by the greedy Frank Wolfe algorithm. Then, given $\beta^{t-1}_i$ for all $i\in [n]$, iteration $t$ of the Frank Wolfe algorithm attempts to find the vector $\beta^t$ such that 
\begin{equation}
\begin{aligned}
\min_{\beta} \quad & \sum_{i \in [n]} B_i \cdot \frac{\beta_i}{\beta^{t-1}_i}\\
\textrm{s.t.} \quad &  p_j - \beta_id_{ij} \leq 0  &\forall i,j\\
  &\quad  \sum_{j \in [m]} p_j =\sum_{i \in [n]}B_i = B    \\
  &\quad p_j \geq 0 &\forall j \\
  &\quad \beta_i \geq 0 &\forall i \\
\end{aligned}
\end{equation}

Let $x^t_{ij}$ be the optimal dual variable corresponding to the constraint $p_j - \beta_i d_{ij} \leq 0$, and $\delta^t$ be the optimal dual variable corresponding to the constraint $\sum_{j \in [m]}p_j = 1$ in the $t$-th iteration of the greedy Frank Wolfe algorithm. 

\begin{claim}
\label{distutility-tech}
    For all $i \in [n]$, we have $\sum_{j \in [m]}d_{ij}x^t_{ij} = B_i/ \beta^{t-1}_i$. 
\end{claim}
\begin{proof}
    Under stationarity conditions of the optimum solutions, we have for all $i \in [n]$, $B_i/ \beta^{t-1}_i = \sum_{j \in [m]}d_{ij} x^t_{ij}$, whenever $\beta^t_i > 0$. Thus, it suffices to show that $\beta^t_i > 0$ for all $i \in [n]$. Assume otherwise, and $\beta^t_{\ell} = 0$. Since $\beta^t_{\ell} = \max_{j \in [m]} p^t_{j}/d_{\ell j}$, this implies that $p^t_j = 0 $ for all $j \in [m]$, implying $\sum_{j \in [m]}p^t_j = 0$, which is a contradiction. 
\end{proof}

We now make another technical observation. 

\begin{claim}
\label{disutilities-outflows}
    For all $i \in [n]$, we have  $\sum_{j \in [m]} d_{ij}x^{t-1}_{ij} = 1/\beta^{t-1}_i \cdot (\sum_{j \in [m]} p^{t-1}_j x^{t-1}_{ij}) $. Analogously, we have  $\sum_{j \in [m]} d_{ij}x^t_{ij} = 1/\beta^t_i \cdot (\sum_{j \in [m]} p^t_j x^t_{ij}) $.
\end{claim}

\begin{proof}
    Under the complementary slackness condition, we have 
    $x^{t-1}_{ij}(p^{t-1}_j - \beta^{t-1}_id_{ij} ) = 0$ for all $i,j$. Therefore, for each $i \in [n]$, we have 
    \begin{align*}
        &\sum_{j \in [m]}x^{t-1}_{ij}(p^{t-1}_j - \beta^{t-1}_id_{ij}) = 0\\
        &\implies \sum_{j \in [m]}p^{t-1}_j x^{t-1}_{ij} = \beta^{t-1}_i \sum_{j \in [m]} d_{ij}x^{t-1}_{ij} \\
        &\implies \sum_{j \in [m]} d_{ij}x^{t-1}_{ij} = 1/\beta^{t-1}_i \cdot (\sum_{j \in [m]} p^{t-1}_j x^{t-1}_{ij}) 
    \end{align*}
    The proof for the second statement of the theorem follows analogously.
\end{proof}
From~\cref{disutilities-outflows}, it follows that $\prod_{i \in [n]} (\sum_{j \in [m]}d_{ij}x^{t-1}_{ij})^{B_i} = (\prod_{i \in [n]} (1/\beta^{t-1}_i) ^{B_i}) \cdot (\prod_{i \in [n]} (\sum_{j \in [m]}$  $p^{t-1}_j x^{t-1}_{ij})^{B_i})$ and~\cref{distutility-tech} it follows that $\prod_{i \in [n]} (\sum_{j \in [m]}d_{ij}x^{t}_{ij})^{B_i} = \prod_{i \in [n]} (B_i/\beta^{t-1}_{i})^{B_i}$ . Therefore, if we comparing the weighted product of disutilities in the current iteration and the previous iteration, we have

\begin{align*}
  &\frac{\prod_{i \in [n]} (\sum_{j \in [m]} d_{ij}x^t_{ij})^{B_i}}{\prod_{i \in [n]} (\sum_{j \in [m]} d_{ij}x^{t-1}_{ij})^{B_i}} \\ =& \frac{\prod_{i \in [n]} (\frac{B_i}{\beta^{t-1}_i})^{B_i}}{(\prod_{i \in [n]} (\frac{1}{\beta^{t-1}_i})^{B_i} ) \cdot (\prod_{i \in [n]} (\sum_{j \in [m]}p^{t-1}_j x^{t-1}_{ij})^{B_i})} \\
  =& \frac{\prod_{i\in [n]} B_i^{B_i}}{\prod_{i \in [n]}(\sum_{j \in [m]}p^{t-1}_j x^{t-1}_{ij})^{B_i}}
\end{align*}

We now make an observation on $\sum_{i \in [n]}x^t_{ij}$-- interpreting the dual variables $x^t_{ij}$ as the fraction of chore $j$ allocated to agent $i$, one can interpret $\sum_{i \in [n]}x^t_{ij}$ as the total consumption of chore $j$ at the end of the $t$-th iteration. 

\begin{claim}
    \label{Allocation-claim}
    For all $i \in [n]$, and for all $t \geq 0$, we have $\sum_{i \in [n]} x^t_{ij} =  \frac{1}{B} \cdot \sum_{i \in [n]} B_i \cdot \frac{\beta^{t}_i}{\beta^{t-1}_i}$. 
\end{claim}

\begin{proof}
    Under stationarity conditions, we have $\sum_{i \in [n]} x^t_{ij} = -\delta^t$ for all $j$ such that $p^t_j > 0$. Furthermore, from~\cref{disutilities-outflows}, it follows that $\sum_{j \in [m]} p^t_j x^t_{ij} = \beta^t_i \cdot (\sum_{j \in [m]} d_{ij} x^t_{ij})$. Observe that,
    \begin{align*}
        \sum_{i \in [n]} B_i \cdot \frac{\beta^t_i}{\beta^{t-1}_i} &=\sum_{i \in [n]} \beta^t_i \cdot \frac{B_i}{\beta^{t-1}_i}\\ 
        &= \sum_{i \in [n]} \beta^t_i \cdot (\sum_{j \in [m]} d_{ij}x^t_{ij})  &\text{(by~\cref{distutility-tech})}\\
        &= \sum_{i \in [n]} \sum_{j \in [m]} p^t_jx^t_{ij} &\text{(by~\cref{disutilities-outflows})}\\
        &= \sum_{j \in [m]} p^t_j \cdot \sum_{i \in [n]} x^t_{ij}\\
        &= \sum_{j \mid p^t_j >0} p^t_j \cdot \sum_{i \in [n]} x^t_{ij}\\
        &=\sum_{j \mid p^t_j >0} p^t_j \cdot (-\delta^t) \\ 
        &= (-\delta^t) \cdot \sum_{j \in [m]} p^t_j\\
        &=- (\sum_{i \in [n]} B_i)\delta^t = -B\delta^t
    \end{align*}
    Thus, we have $-\delta^t = \frac{1}{B} \cdot \sum_{i \in [n]} B_i \cdot \frac{\beta^t_i}{\beta^{t-1}_i}$. Therefore, for all $j \in [m]$, such that $p^t_j > 0$, we have $\sum_{i \in [n]} x^t_{ij} = - \delta^t = \frac{1}{B} \cdot \sum_{i \in [n]} B_i \cdot \frac{\beta^t_i}{\beta^{t-1}_i}$. It only suffices to show that $p^t_j > 0$ for all $j \in [m]$.
    
    We prove this by contradiction. Assume there exists a $j \in [m]$, such that $p^t_j = 0$. Then note that $p^t_j/d_{i j} < \beta^t_{i}$ for all $i \in [n]$ (follows from the fact that every $\beta^t_i$ is strictly positive). Note that there exists a $\gamma \ll  \min_{\ell \in [n]} d_{\ell j} \beta^t_{\ell}$,  such that if we increase $p^t_j$ to $\gamma$, and decrease all other prices by a multiplicative factor of $(B-\gamma)$, we still have a feasible set of prices\footnote{All the prices are still non-negative and they still sum up to $B$.} and $\max_{j \in [m]} p^t_j/d_{ij}$ strictly decreases for all $i \in [n]$. Let $\beta'_i = \max_{j \in [m]} p^t_j/d_{ij}$ (post the multiplicative scaling). Note that $\beta'_i < \beta^t_i$ for all $i \in [n]$ and thus $\sum_{i \in [n]} B_i \cdot \beta'_i/\beta^{t-1}_i <\sum_{i\in [n]} B_i \cdot \beta^t_i/\beta^{t-1}_i$ which contradicts the optimality of $\beta^t$.
\end{proof}

\begin{claim}
\label{progress}
If $( p^{t-1}, x^{t-1} )$ is not an $\epsilon$-approximate equilibrium, then $\log \Bigg(\frac{\prod_{i\in [n]} B_i^{B_i}}{\prod_{i \in [n]}(\sum_{j \in [m]}p^{t-1}_j x^{t-1}_{ij})^{B_i}} \Bigg) \geq \min_i B_i \cdot \epsilon^2/2.3$. 
\end{claim}

\begin{proof}
Note that $x^{t-1}_{ij} > 0$ only if $d_{ij}/p_j$ is minimum for all $j \in [m]$ for agent $i$. Thus, if $( p^{t-1}, x^{t-1} )$ is not an $\epsilon$-approximate equilibrium, then there exists a $j \in [m]$ such that $\sum_{i \in [n]} x^{t-1}_{ij} < 1 - \epsilon$, or there exists an $i \in [n]$, such that $|\sum_{j \in [m]} p^{t-1}_jx^{t-1}_{ij} - B_i| >  B_i\epsilon$.

   \begin{itemize}
       \item \textbf{Case $\sum_{i \in [n]} x^{t-1}_{ij} < 1 - \epsilon$:} From~\cref{Allocation-claim}, it follows that $\sum_{i \in [n]} x^{t-1}_{ij} = \frac{1}{B} \cdot \sum_{i \in [n]} B_i\frac{\beta^{t-1}_i}{\beta^{t-2}_i} < 1 - \epsilon$, implying that $\sum_{i \in [n]} B_i \frac{\beta^{t-1}_i}{\beta^{t-2}_i} < B(1 - \epsilon)$. Further, observe that
       \begin{align*}
        \sum_{j \in [n]} p^{t-1}_j \cdot x^{t-1}_{ij} &= \beta^{t-1}_i \cdot \sum_{j \in [m]} d_{ij}x^{t-1}_{ij} & \text{(by~\cref{disutilities-outflows})}\\
        &= \beta^{t-1}_i \cdot \frac{B_i}{\beta^{t-2}_i} & \text{(by~\cref{distutility-tech})}
    \end{align*}
    Thus, we have $\sum_{i \in [n]} \sum_{j \in [m]} p^{t-1}_j x^{t-1}_{ij} = \sum_{i \in [n]} B_i \frac{\beta^{t-1}_i}{\beta^{t-2}_i} < B(1-\epsilon)$. This implies that $\prod_{i \in [n]} (\sum_{j \in [m]} $ $ p^{t-1}_j x^{t-1}_{ij})^{B_i} \leq \prod_{i \in [n]} (B_i(1-\epsilon))^{B_i}$, further implying that 
    \begin{align*}
        \frac{\prod_{i \in [n]} B_i^{B_i}}{\prod_{i \in [n]} (\sum_{j \in [m]} p^{t-1}_j x^{t-1}_{ij})^{B_i}} &\geq \frac{1}{(1-\epsilon)^B}\\
                                                                   &\geq (1+\epsilon)^B\\
                                                                   &\geq 1 + {B\epsilon }
    \end{align*}
    Therefore, $\log \Bigg( \frac{\prod_{i \in [n]} B_i^{B_i}}{\prod_{i \in [n]} (\sum_{j \in [m]} p^{t-1}_j x^{t-1}_{ij})^{B_i}} \Bigg) \geq \log(1+B\epsilon) \geq \min_i B_i \epsilon^2/2.3$ as $\log(1+x) \geq x^2/2.3$ for all $x \in [0,1]$.

      \item \textbf{Case $|\sum_{j \in [m]} p^{t-1}_jx^{t-1}_{ij} - B_i| >  B_i\epsilon$:}  
      Let $i$ be an agent such that $\sum_{j \in [m]} p^{t-1}_{j}x^{t-1}_{i j} = B_i(1 + \epsilon')$, where $\epsilon' \in \{-\theta, \theta \}$ for some $\theta > \epsilon$. Observe that $\sum_{\ell \in [n] \setminus \{i\}} (\sum_{j \in [m]} p^{t-1}_jx^{t-1}_{ \ell j}) < B- B_i-B_i\epsilon'$, as $\sum_{i \in [n]}(\sum_{j \in [m]} p^{t-1}_jx^{t-1}_{ij}) = \sum_{i \in [n]} B_i\frac{\beta^{t-1}_i}{\beta^{t-2}_i} \leq B$, and $\sum_{j \in [m]} p^{t-1}_{j}x^{t-1}_{i j} = B_i(1 + \epsilon')$. This implies that $\prod_{\ell \in [n] \setminus \{i\}} (\sum_{j \in [m]} p^{t-1}_j x^{t-1}_{ \ell j})^{B_i} \leq \prod_{\ell \in [n] \setminus \{i\}} (B_{\ell} \cdot (1-\frac{B_i\epsilon'}{B-B_i}))^{B_{\ell}}$. Before, we proceed, we mention the following useful fact,

      \begin{fact}
          \label{technical}
           For all real $x$, we have, $\log(e^x/(1+x)) \geq x^2/2.3$ for all $x \in [-1,1]$.
      \end{fact}

      \begin{proof}
      We have,
      \begin{align*}
          & \hspace{20pt} \frac{\prod_{i \in [n]} B_i^{B_i}}{\prod_{i \in [n]} (\sum_{j \in [m]} p^{t-1}_j x^{t-1}_{ij})^{B_i}} \\ 
          &> \frac{\prod_{i \in [n]} B_i^{B_i}}{(B_i(1+\epsilon'))^{B_i} \prod_{\ell \in [n] \setminus \{i\}} (B_{\ell} \cdot (1 - \frac{B_i\epsilon'}{B-B_i}))^{B_{\ell}}}\\
          &=\frac{1}{(1 + \epsilon')^{B_i} (1-\frac{B_i\epsilon'}{B-B_i})^{(B-B_i)}}\\
          &\geq \frac{1}{(1 + \epsilon')^{B_i} e^{-\frac{B_i\epsilon'}{B-B_i} \cdot (B-B_i)}}\\  &= \frac{1}{(1 + \epsilon')^{B_i}e^{-B_i\epsilon'}}\\
          &=\bigg( \frac{1}{(1+\epsilon')e^{-\epsilon'}} \bigg)^{B_i} 
      \end{align*}
      \end{proof}
   \end{itemize}
\end{proof}
Now, observe that \
\begin{align*}
     \log \Bigg( \frac{\prod_{i \in [n]} B_i^{B_i}}{\prod_{i \in [n]} (\sum_{j \in [m]} p^{t-1}_j x^{t-1}_{ij})^{B_i}} \Bigg) &> \log \Bigg( \bigg(\frac{1}{(1+\epsilon')e^{-\epsilon'}}\bigg)^{B_i} \Bigg)\\
     &=B_i \log \Bigg(\frac{e^{\epsilon'}}{1+\epsilon'} \Bigg)\\
     &\geq B_i \frac{\epsilon'^2}{2.3} &\text{(from Fact~\ref{technical})}\\
     &= \min_i B_i \frac{\theta^2}{2.3} = \min_i B_i \frac{\epsilon^2}{2.3}. 
\end{align*}

We are now ready to prove our main result.

\begin{theorem}
\label{Theorem-convergence}
    Let $T \geq \frac{2.3(f(\beta^*) - f(\beta^0))}{\epsilon^2 \min_i B_i }$. Then there exists a $t \leq T$, such that $(p^t,x^t)$ is a $\epsilon$-CE 
\end{theorem}

\begin{proof}
    Assume otherwise, i.e., for all $t \in [T]$,  $(p^t,x^t)$ is not a $\epsilon$-CE. Recall that we have,
    \begin{align*}
        \frac{\prod_{i \in [n]} (\sum_{j \in [m]} d_{ij}x^t_{ij})^{B_i}}{\prod_{i \in [n]} (\sum_{j \in [m]} d_{ij}x^{t-1}_{ij})^{B_i}} = \frac{\prod_{i \in [n]} B_i^{B_i} }{\prod_{i \in [n]}(\sum_{j \in [m]}p^{t-1}_j x^{t-1}_{ij})^{B_i}}
    \end{align*}
   By~\cref{distutility-tech}, we can replace  $\sum_{j \in [m]} d_{ij}x^t_{ij}$ with $1/\beta^{t-1}_i$, implying,
   \begin{align}
       \label{eqn:prod-beta-comp}
        \frac{\prod_{i \in [n]} (1/\beta^{t-1}_i)^{B_i}}{\prod_{i \in [n]} (1/\beta^{t-2}_i)^{B_i}} = \frac{\prod_{i \in [n]} B_i^{B_i} }{\prod_{i \in [n]}(\sum_{j \in [m]}p^{t-1}_j x^{t-1}_{ij})^{B_i}}
    \end{align}
   Since $(p^{t-1},x^{t-1})$ is not a $\epsilon$-CEEI, we have $ \log \bigg( \frac{\prod_{i \in [n]} B_i^{B_i} }{\prod_{i \in [n]}(\sum_{j \in [m]}p^{t-1}_j x^{t-1}_{ij})^{B_i}} \bigg) \geq \min_i B_i \cdot \epsilon^2/2.3$ by~\cref{progress}. Taking logs on both sides of~\cref{eqn:prod-beta-comp}, we have
   \begin{align}
      \label{eq:beta-teslescoping}
       f(\beta^{t-1}) - f(\beta^{t-2})  \geq  \min_i B_i \cdot \epsilon^2/2.3
   \end{align}
   Summing up~\cref{eq:beta-teslescoping} from $t=2$ to $T+1$, we have $f(\beta^T) - f(\beta^0) > \min_i B_i \cdot T\epsilon^2/2.3$. Since $f(\beta^*) \geq f(\beta^T)$, this implies that $f(\beta^*) - f(\beta^0) > \min_i B_i \cdot T\epsilon^2/2.3$, implying that $T < 2.3(f(\beta^*) - f(\beta^0))/(\epsilon^2 \cdot \min_i B_i)$, which is a contradiction. 
\end{proof}

\begin{remark}
     This analysis also gives us additional insights: for instance, any iteration where the solution of the GFW is not a good approximation of a CE, will see a large improvement in the objective function, implying that GFW iterations with bad approximate solutions are small. 
\end{remark}

\label{sec:polytime-conv-approxi}

\section{Approximate CE and Exact CE can be far Apart}
\label{example-CE}
 We show that approximate CE and Exact CE can be far apart. Consider the following chores market: 2 agents $a_1$ and $a_2$, and 2 chores $c_1$ and $c_2$. Set $d_{11} = 1$, and $d_{12} = M$ for $M \gg 1$. Set $d_{21} = 1-\epsilon$, and $d_{22} = 1 + \epsilon$ for $\epsilon \ll 1$. It is easy to verify that the only CEEI here is  
\begin{itemize}
    \item $p_1 = 2/(M+1)$ and $p_2 = 2M/(M+1)$,
    \item $x_{11} = 1$, $x_{21} = (M-1)/(2M)$, and 
    \item $x_{21} = 0$, and $x_{22} = (M+1)/(2M)$.
\end{itemize}

 Note that $\beta_1 = 2/(M+1)$, and $\beta_2 = 2M/((M+1)(1+\epsilon))$. However, the following is an $\epsilon$-approximate CEEI, with $\beta_1=\beta_2 = 1$, which is far from the only exact CEEI!

 \begin{itemize}
    \item $p_1 = 1-\epsilon$ and $p_2 = 1+\epsilon$,
    \item $x_{11} = 1-\epsilon$, $x_{21} = 0$, and 
    \item $x_{21} = 0$, and $x_{22} = 1$.
\end{itemize}

\section{Additional numerical experiments}
\label{sec:appendix experiments}

\begin{figure}[H]
    \centering
    \makebox[0pt][r]{\makebox[30pt]{\raisebox{40pt}{\rotatebox[origin=c]{90}{uniform(0,1)}}}}%
    \includegraphics[width=0.28\linewidth]{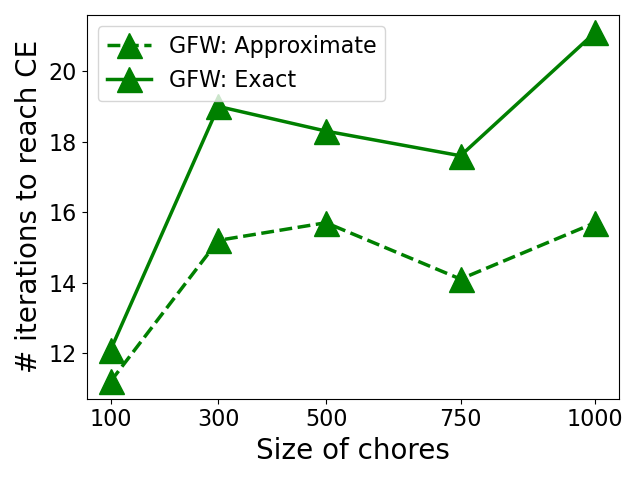}
    \includegraphics[width=0.28\linewidth]{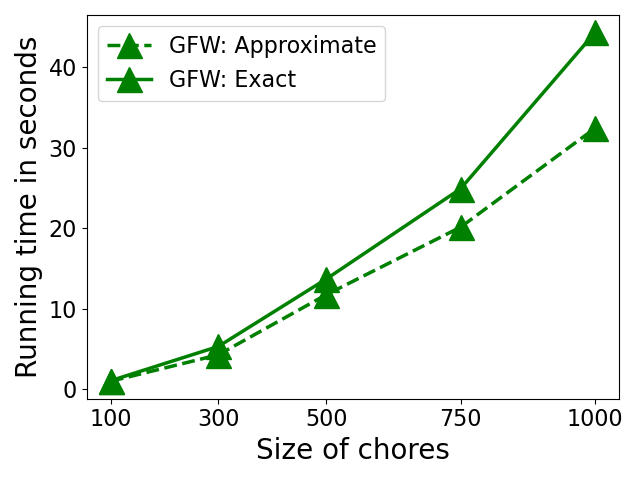}
    \includegraphics[width=0.28\linewidth]{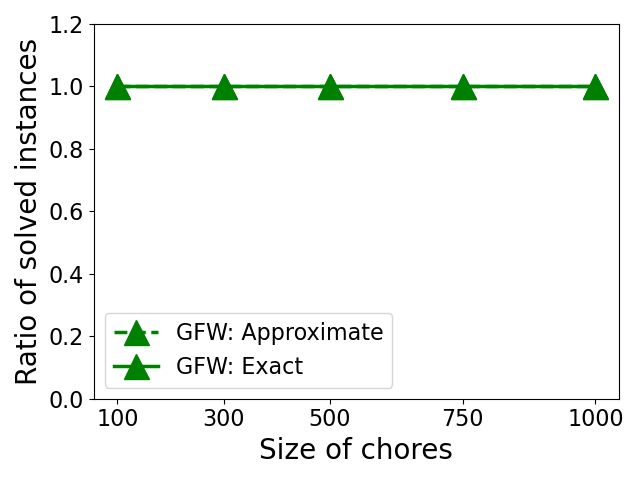}
    \makebox[0pt][r]{\makebox[30pt]{\raisebox{40pt}{\rotatebox[origin=c]{90}{lognormal}}}}%
    \includegraphics[width=0.28\linewidth]{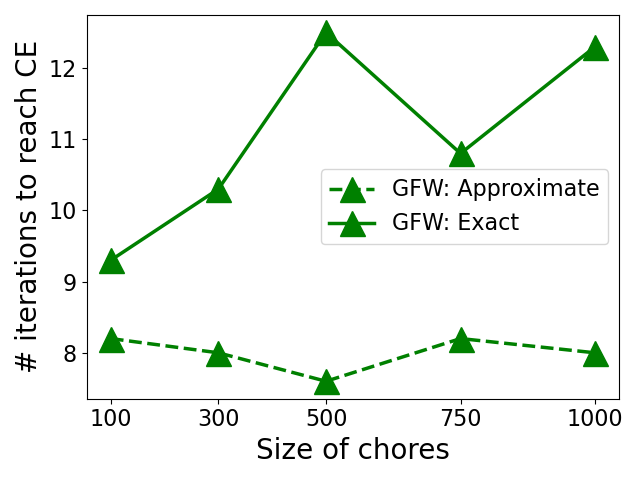}
    \includegraphics[width=0.28\linewidth]{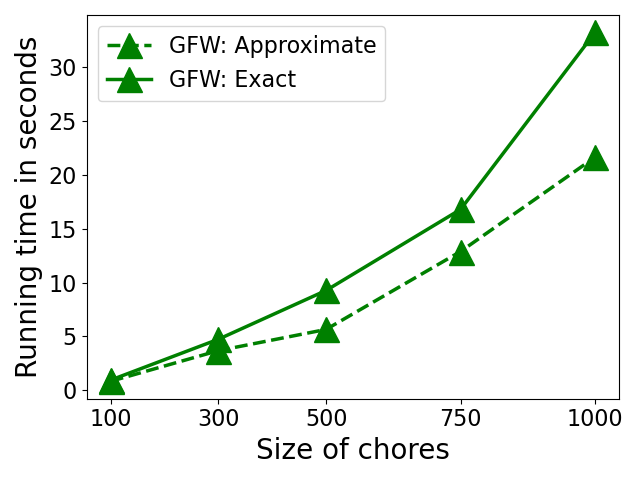}
    \includegraphics[width=0.28\linewidth]{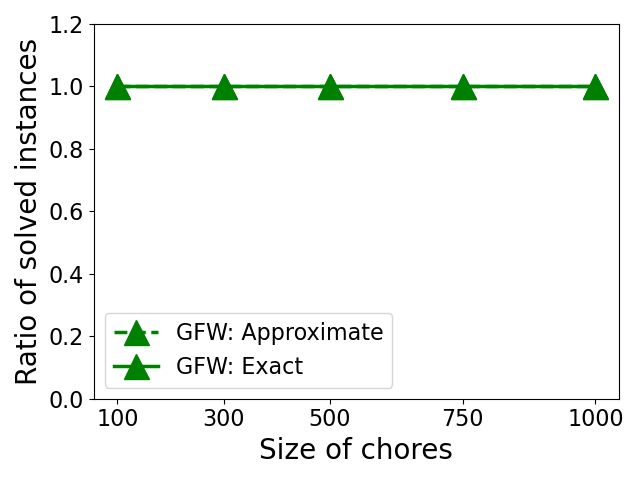}
    \makebox[0pt][r]{\makebox[30pt]{\raisebox{40pt}{\rotatebox[origin=c]{90}{truncated normal}}}}%
    \includegraphics[width=0.28\linewidth]{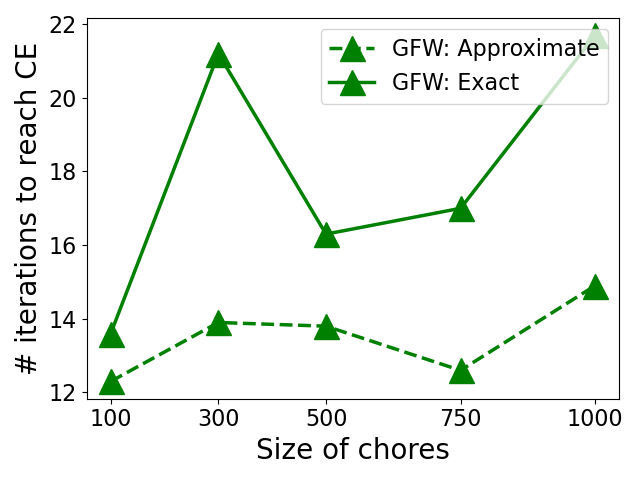}
    \includegraphics[width=0.28\linewidth]{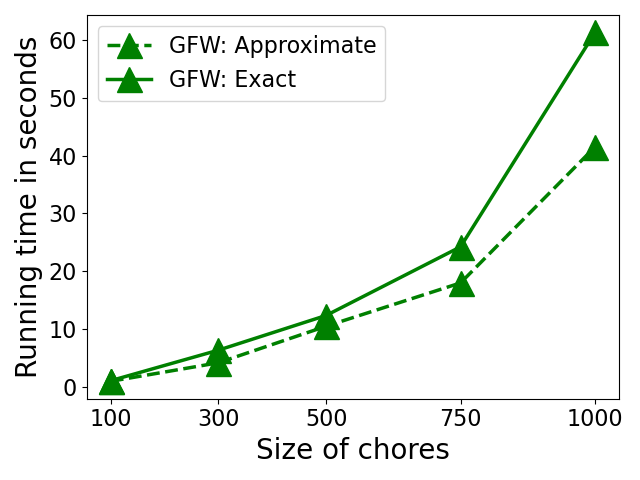}
    \includegraphics[width=0.28\linewidth]{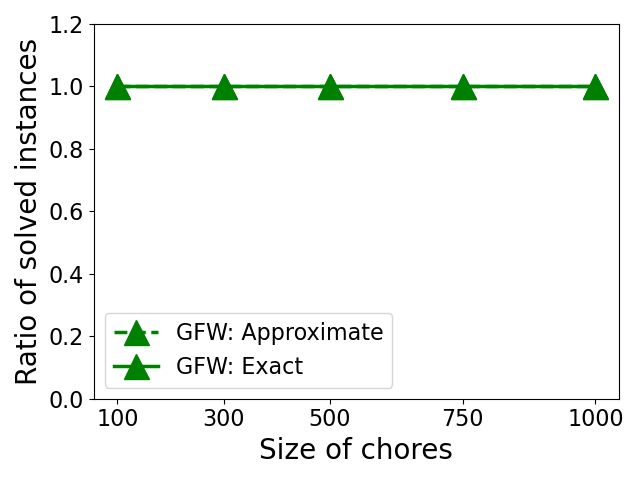}
    \makebox[0pt][r]{\makebox[30pt]{\raisebox{40pt}{\rotatebox[origin=c]{90}{exponential}}}}%
    \includegraphics[width=0.28\linewidth]{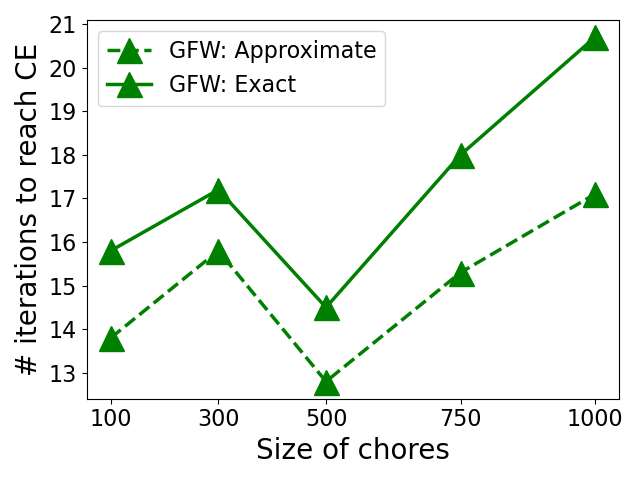}
    \includegraphics[width=0.28\linewidth]{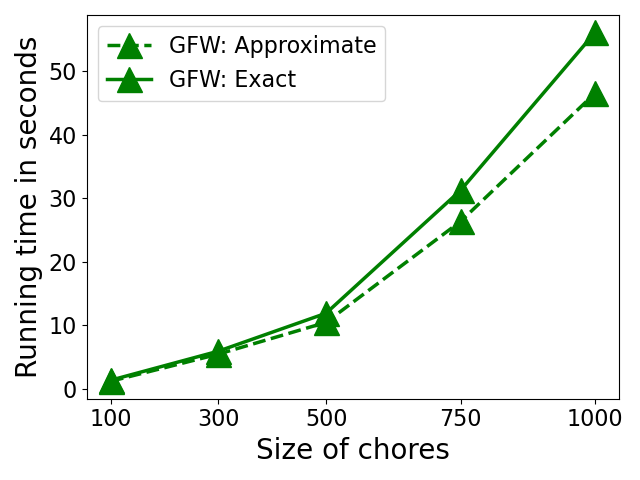}
    \includegraphics[width=0.28\linewidth]{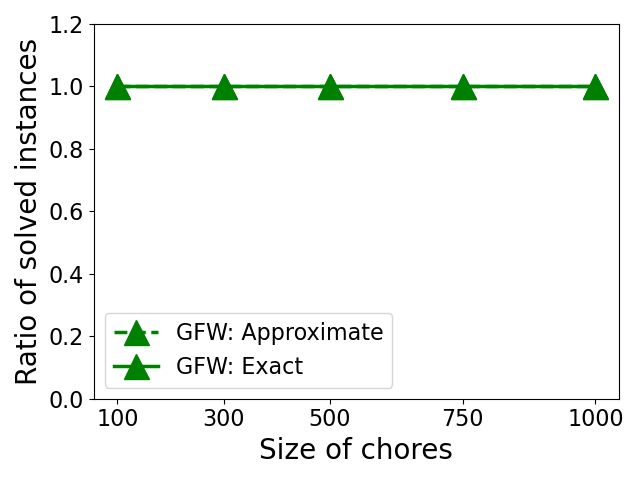}
    \makebox[0pt][r]{\makebox[30pt]{\raisebox{40pt}{\rotatebox[origin=c]{90}{randint(1,1000)}}}}%
    \includegraphics[width=0.28\linewidth]{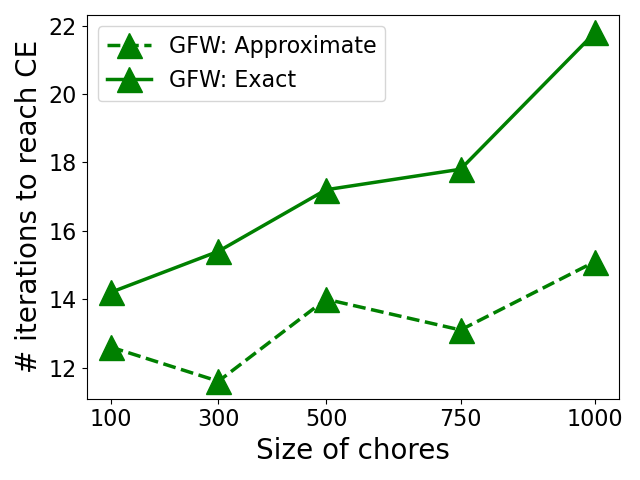}
    \includegraphics[width=0.28\linewidth]{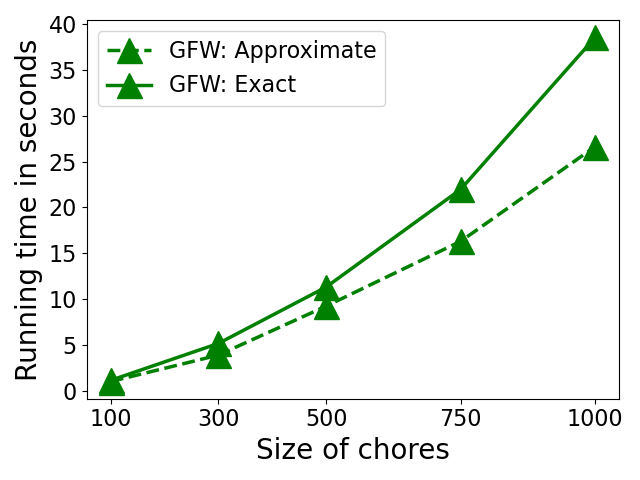}
    \includegraphics[width=0.28\linewidth]{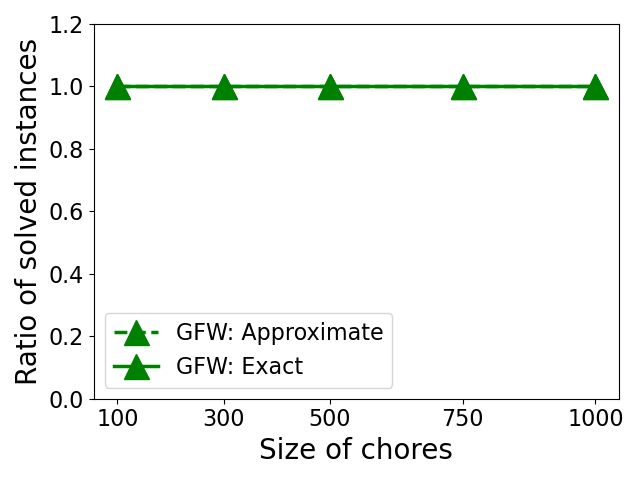}
    \caption{
    Numerical results for 100 buyers and up to 1000 chores.
    GFW solves every instance for every distribution. 
    }
    \label{fig:random chores large size instances}
\end{figure}

\begin{figure}[H]
    \centering
    \makebox[0pt][r]{\makebox[30pt]{\raisebox{40pt}{\rotatebox[origin=c]{90}{uniform(0,1)}}}}%
    \includegraphics[width=0.28\linewidth]{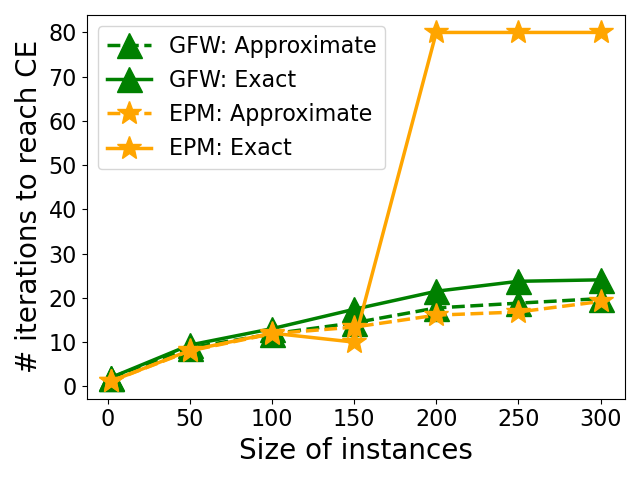}
    \includegraphics[width=0.28\linewidth]{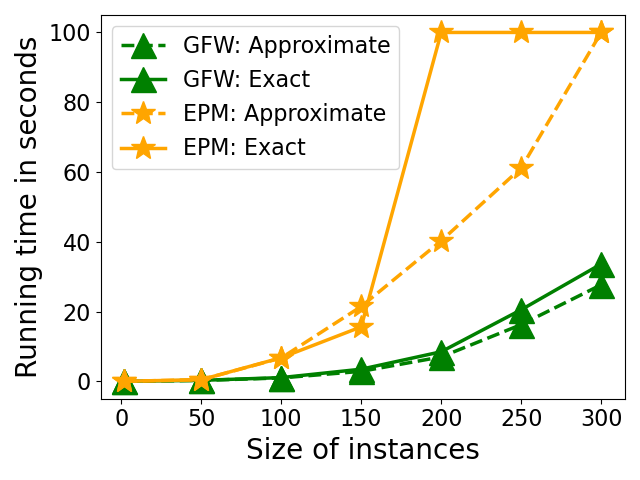}
    \includegraphics[width=0.28\linewidth]{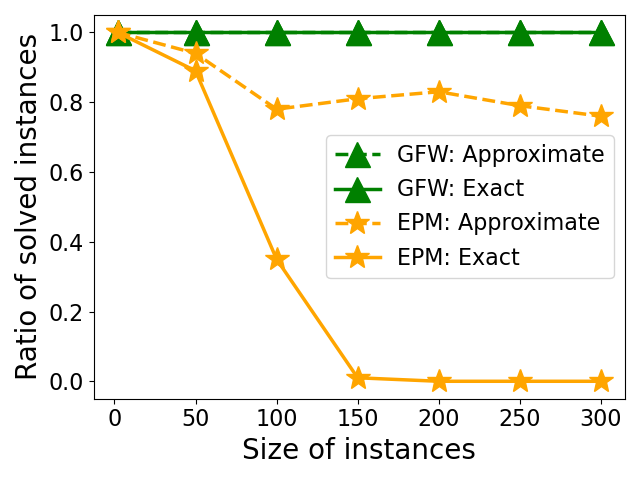}
    \makebox[0pt][r]{\makebox[30pt]{\raisebox{40pt}{\rotatebox[origin=c]{90}{lognormal}}}}%
    \includegraphics[width=0.28\linewidth]{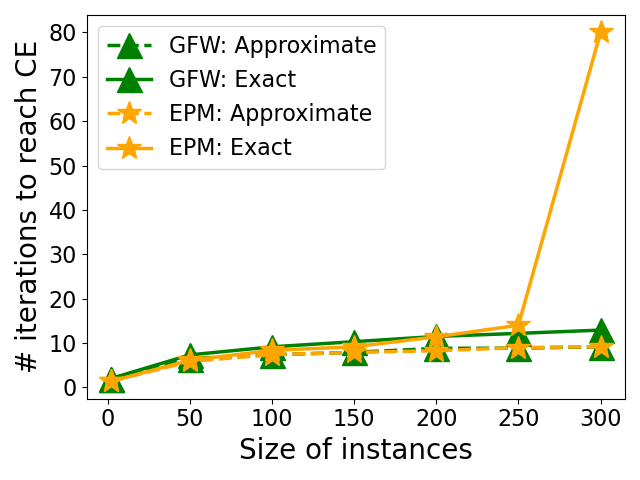}
    \includegraphics[width=0.28\linewidth]{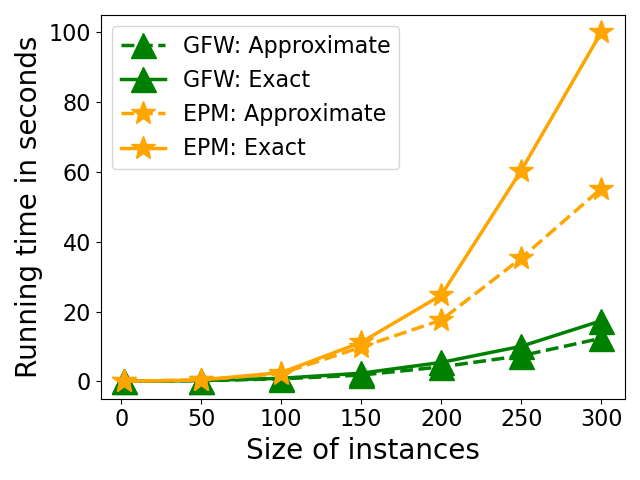}
    \includegraphics[width=0.28\linewidth]{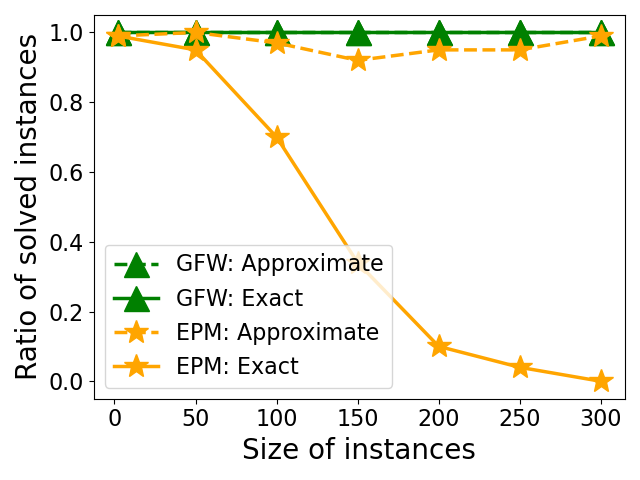}
    \makebox[0pt][r]{\makebox[30pt]{\raisebox{40pt}{\rotatebox[origin=c]{90}{truncated normal}}}}%
    \includegraphics[width=0.28\linewidth]{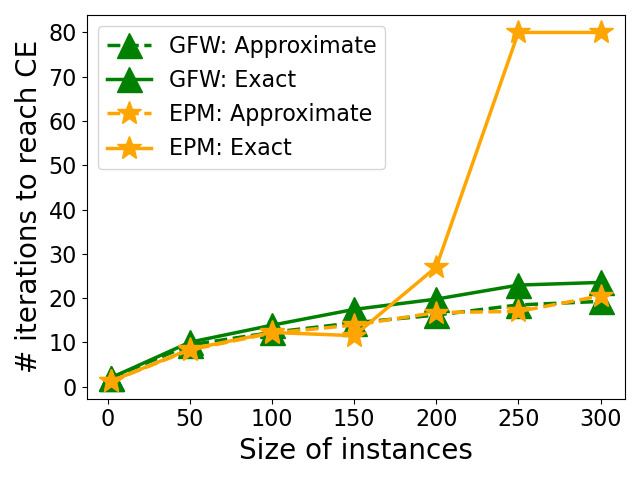}
    \includegraphics[width=0.28\linewidth]{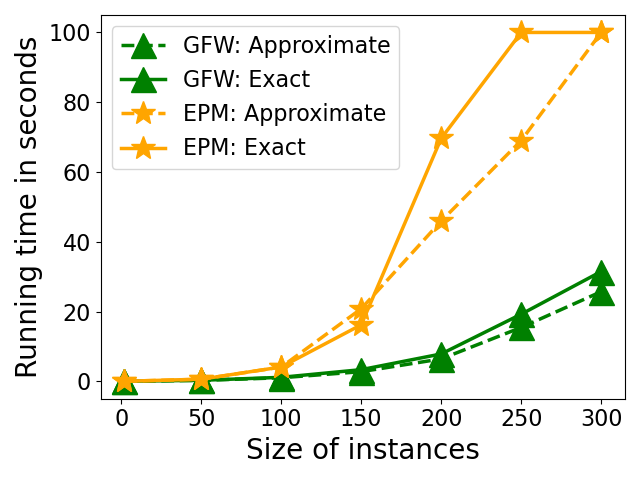}
    \includegraphics[width=0.28\linewidth]{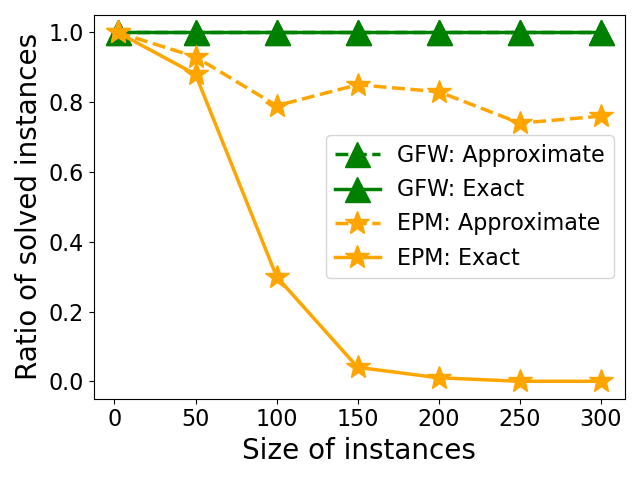}
    \makebox[0pt][r]{\makebox[30pt]{\raisebox{40pt}{\rotatebox[origin=c]{90}{exponential}}}}%
    \includegraphics[width=0.28\linewidth]{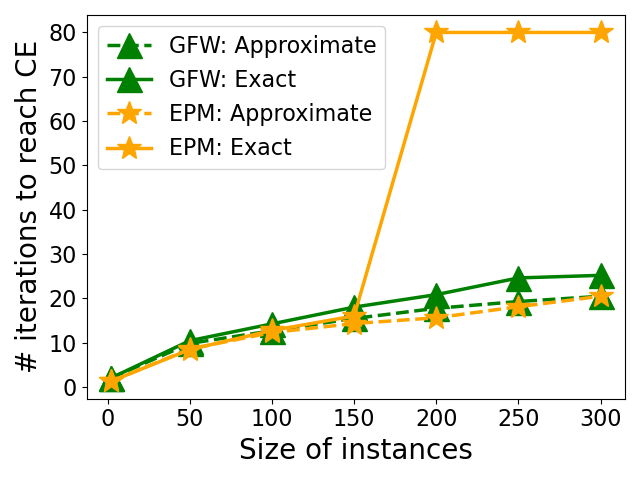}
    \includegraphics[width=0.28\linewidth]{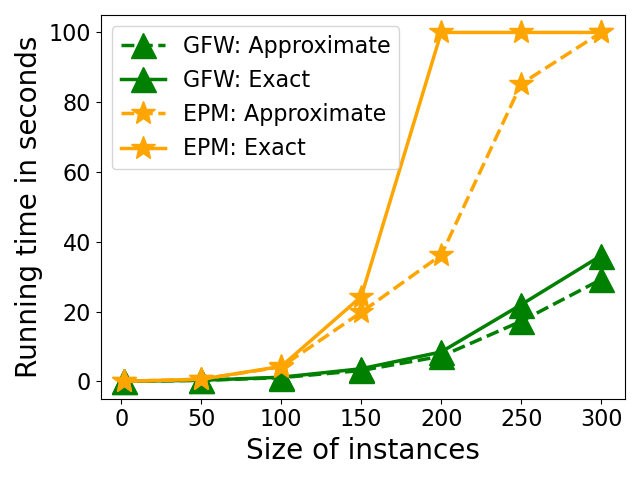}
    \includegraphics[width=0.28\linewidth]{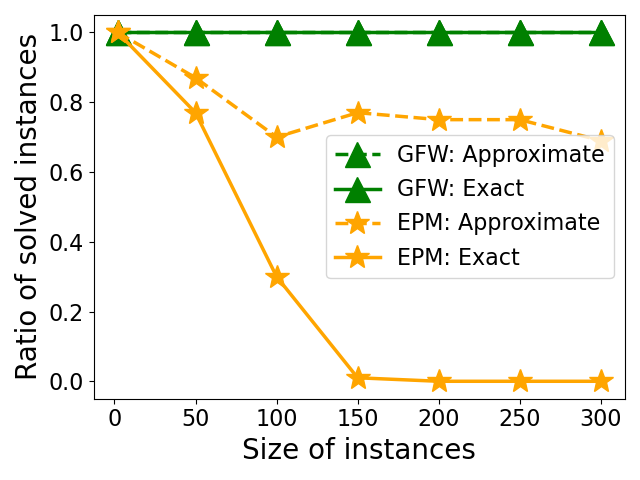}
    \makebox[0pt][r]{\makebox[30pt]{\raisebox{40pt}{\rotatebox[origin=c]{90}{randint(1,1000)}}}}%
    \includegraphics[width=0.28\linewidth]{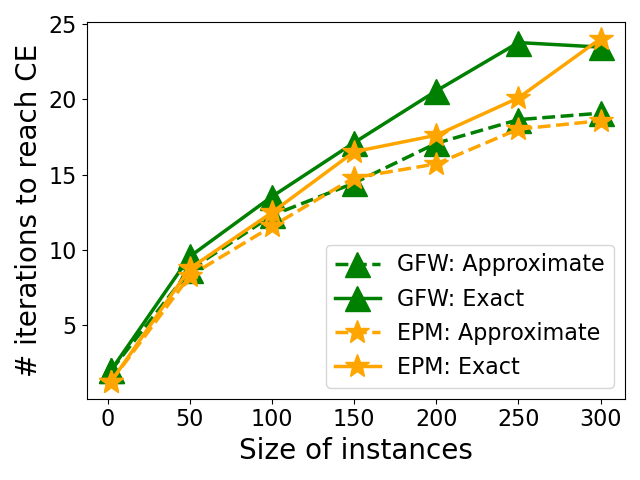}
    \includegraphics[width=0.28\linewidth]{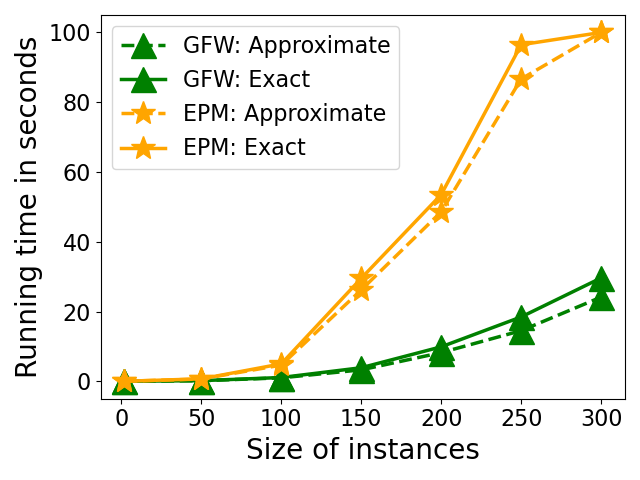}
    \includegraphics[width=0.28\linewidth]{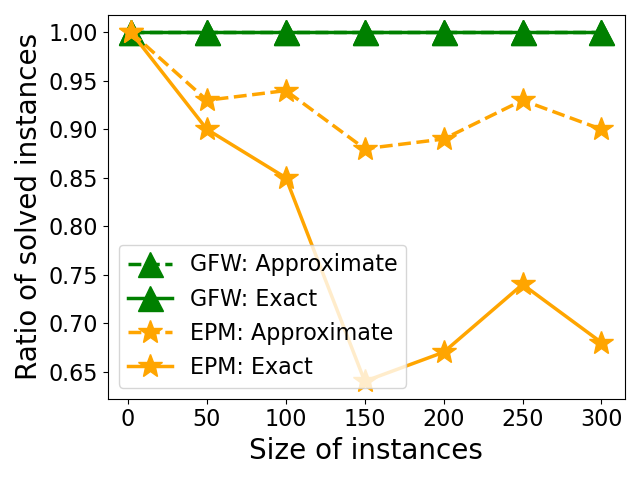}
    \caption{
    Numerical results with very small Gurobi tolerance parameters.}
    \label{fig:random chores low tolerance parameters}
\end{figure}

\end{document}